%% file: main.tex
\documentclass[onefignum,onetabnum]{siamsials251208}
\usepackage{graphicx}
\usepackage{tikz}
\usepackage{tikz-cd}
\usetikzlibrary{arrows.meta}
\usepackage{enumitem}
\usepackage{url}
\usepackage{algorithm}
\usepackage{algpseudocode}
\usepackage{amsmath}
\usepackage{multirow}

\usepackage{amsfonts}
\usepackage{graphics}
\usepackage[
 letterpaper, top=1.2in, bottom=1.2in, left=1.2in, right=1.2in]{geometry}
\usepackage{xcolor}
\usepackage{amssymb}
\usepackage{psfrag}
\usepackage[utf8]{inputenc}
\usepackage{bm}
\usepackage{mathtools}
\usepackage{epstopdf}
\usepackage{setspace}
\usepackage{amsmath,amssymb,bbm}
\usepackage[utf8]{inputenc}
\usepackage{url}
\usepackage{color}
\usepackage{booktabs}

\input{ex_shared}

\usepackage{bbm}

\allowdisplaybreaks

\begin{document}
\maketitle
\begin{abstract}
Predicting how tightly two biomolecules bind remains a major challenge, in part because different interaction classes present dissimilar interfaces, from compact metal-coordinated pockets to broad, featureless protein surfaces. We introduce persistent manifold learning (PML), a novel computational framework that describes a binding interface as a family of multiscale manifolds. Boundary-Induced Graph Laplacian, a discrete realization of de Rham--Hodge theory, then extracts topological invariants together with nonharmonic spectral information, capturing the geometry of an interface as well as its topology. These manifold embeddings are combined with protein and molecular language model representations and paired with gradient boosting decision trees. Our PML 
outperforms state-of-the-art methods on metalloprotein--ligand and protein--protein benchmarks. 
\end{abstract}

\begin{relevance}
This work addresses a central problem in molecular biology and drug discovery: how to determine, from molecular structure alone, how tightly a molecule will bind to its target. Binding affinity governs therapeutic potency and selectivity, yet measuring it experimentally is slow and costly, making computational prediction essential for prioritizing candidates from large libraries. We apply our framework to two target classes ordinarily handled by separate methods: metalloenzymes, whose metal-coordinated active sites have long been exploited by drugs, and protein--protein interfaces, which regulate signaling and immune response yet resist conventional small molecules. Our results indicate that much of what determines binding strength is encoded in the shape of the interface itself, and that a single geometric description serves both classes without hand-tailored features. This suggests manifold-based representations may extend to other systems in which structure and sequence jointly determine function.
\end{relevance}

\begin{mathcontent}
The core of our approach involves de Rham--Hodge theory on compact Riemannian manifolds with boundary, a foundational tool in differential geometry rarely applied to molecular data. Binding interfaces are modeled as sublevel sets of Gaussian density fields, producing a filtration of manifolds whose topological transitions occur at the critical values of the level set function, as governed by Morse theory. The topology of each manifold is recovered from the kernels of Hodge Laplacians under normal and tangential boundary conditions, whose dimensions are Betti numbers by the Hodge and Friedrichs theorems, while their nonzero spectra reflect geometry. Discretization proceeds through discrete exterior calculus on Cartesian grids, yielding Boundary-Induced Graph Laplacians whose spectral analysis reduces to the singular values of the discrete differentials. Extending this across the filtration leads to persistent Hodge Laplacians on evolving manifolds, combining differential geometry, algebraic topology, and spectral theory into a multiscale representation of molecular shape.
\end{mathcontent}



\tableofcontents

\newpage
 
\section{Introduction}
Proteins are the basic building blocks of living things, forming the structural and functional core of every cell in the human body. The understanding of protein properties is the essential task of biophysics and molecular biology. A major class of protein properties is protein interactions, which are central to nearly all biological processes, and their dysregulation is closely associated with the onset and progression of many diseases. Understanding how proteins recognize and bind their molecular partners has therefore become an important problem in modern drug discovery~\cite{du2016protein,lu2020protein,nada2024protein,zhao2022brief}. Among the many forms such interactions can take, metalloprotein–ligand interactions (MPLIs), in which small-molecule ligands bind to proteins containing metal ions or metal-coordination environments, and protein–protein interactions (PPIs), in which two proteins interact directly, represent two biologically important classes. MPLIs are a specialized but important subclass of protein-ligand interactions (PLIs), where binding is influenced not only by conventional noncovalent contacts but also by metal coordination, local geometry, charge distribution, and the surrounding protein environment~\cite{bennett2022protein,du2016protein,zhao2022brief}. These metal-centered binding sites provide structurally organized regions for ligand recognition, but they also introduce additional complexity due to coordination preferences, heterogeneous chemical environments, molecular flexibility, solvent effects, and conformational dynamics. In contrast, PPIs often involve broader and flatter interfaces that lack compact ligand-binding pockets. These surface-mediated contacts regulate essential cellular processes such as signaling, immune response, and transcriptional control~\cite{lu2020protein,scott2016small}, while the absence of well-defined pockets makes them difficult to modulate using the small-molecule ligands commonly designed for conventional PLIs.

Despite these structural and mechanistic differences, MPLIs and PPIs are ultimately governed by a common thermodynamic descriptor: binding affinity (BA), which reflects complex stability, target engagement, therapeutic potency, and binding selectivity~\cite{du2016protein,gilson1997statistical}. Since the experimental determination of BA can be costly and time-consuming, computational prediction has emerged as an efficient strategy for screening large molecular spaces and prioritizing promising candidates~\cite{kitchen2004docking}. Such predictions, in turn, rely on informative features that encode how molecular partners interact, enabling complexes to be compared, ranked, and optimized in drug discovery and biological design~\cite{erijman2014structure,vangone2015contacts}. As a free-energy quantity associated with complex formation, BA is strongly shaped by the three-dimensional organization of the binding region, including contact patterns, physicochemical interactions, spatial complementarity, and local packing. In PPIs, this organization is primarily determined by the extended protein-protein interface itself, whereas in MPLIs, it is further shaped by metal-coordination geometry at the binding site. Therefore, a central challenge in BA prediction is to construct molecular representations that preserve binding-relevant structural information across both metal-centered ligand-binding sites and protein-protein interfaces.

For this purpose, existing BA prediction methods often use graph-based descriptors, which model molecular components as nodes and their interactions as edges~\cite{feinberg2018potentialnet,jiang2021interactiongraphnet,li2020monn,moesser2022protein,nguyen2021graphdta}. These descriptors are effective for encoding local contacts and pairwise interaction patterns, but they may not fully describe higher-order structural organization. Since such descriptors serve as the input to downstream machine learning models, the structural information they retain sets an upper bound on predictive accuracy, making descriptor design as consequential as the choice of learning algorithm. Algebraic topology, or persistent homology,  offers a complementary perspective by characterizing this higher-order organization directly, and its combination with machine learning has become an established paradigm for molecular property prediction \cite{townsend2020representation}. In particular, persistent homology captures the appearance and disappearance of connected components, loops, and cavities across a filtration, and the resulting topological invariants have been paired with deep networks and tree-based models for BA prediction~\cite{cang2018representability,cang2017topologynet,cang2018integration,long2025interpretable,xia2014persistent}. However, this approach is insensitive to non-topological changes. Persistent spectral graph (PSG), also known as persistent Laplacian,  was introduced by Wang et al. to overcome this challenge in 2020 \cite{wang2020persistent}. This new method further enriches topological analysis by incorporating spectral information associated with these non-topological changes \cite{liu2024algebraic,memoli2022persistent}.  It has been coupled with machine learning for biomolecular modeling~\cite{chen2022persistent,memoli2022persistent,meng2021persistent,xu2024pldtree}.  
Persistent path Laplacian was proposed to deal with directed networks \cite{wang2023persistentpath}. 
Persistent sheaf Laplacian was introduced to extract  
localized topological information \cite{wei2025persistent}. 
However, these approaches are for point cloud data and cannot be directly applied to data on manifolds. 
The reader is referred to \cite{su2025topological} for a review of topological data analysis and related developments.

The above-mentioned limitation calls for new mathematical modeling.    
Evolutionary de Rham--Hodge theory provides a natural mathematical foundation for this viewpoint. It constructs multiscale manifolds for topological analysis and defines Hodge Laplacian operators on these evolving domains~\cite{chen2021evolutionary}. Building on this framework, persistent de Rham--Hodge Laplacians were also developed in the  Eulerian representation. This approach overcomes the difficulty of the remeshing error due to the Lagrangian representation for manifold topological learning and is applied to image analysis \cite{liu2026manifold} and protein--ligand binding affinity prediction~\cite{su2024persistent}, providing a principled route toward persistent manifold learning (PML), where topology, geometry, and spectral information are integrated for BA prediction.

Building on this foundation, the objective of this work is to extend and demonstrate the utility of PML for BA prediction across two representative systems: metalloprotein–ligand complexes and protein–protein complexes. In the PML framework, molecular structures are first transformed into multiscale manifold representations using Gaussian density functions; persistent Hodge Laplacian descriptors are then computed on level-set domains generated with different isovalues, allowing the resulting features to jointly encode topological persistence and geometric variation. These approaches are subsequently combined with machine learning models to predict binding affinities. We evaluate the proposed approach on two benchmark datasets: the metalloprotein–ligand binding dataset~\cite{jiang2023metalprognet} and the wild-type PDB subset from SKEMPI v2 for protein–protein binding~\cite{jankauskaite2019skempi}. Across both benchmarks, the proposed PML model achieves state-of-the-art performance, demonstrating its effectiveness and generality for BA prediction in structurally distinct binding systems.

The remainder of this article is structured as follows. Section~\ref{PML_framework} introduces the PML framework for BA prediction. Section~\ref{results} presents results on Metalloprotein-ligand and protein-protein BA prediction. Section~\ref{Methods} lays out the mathematical foundations of PML, centered on de Rham-Hodge theory and persistent de Rham-Hodge Laplacians. Section~\ref{Conclusion} concludes with a summary of findings and future directions.

\section{Modeling}\label{Methods}

This section presents the mathematical foundations underlying the PML framework introduced in Section~\ref{PML_framework}. We first review the de Rham-Hodge theory, which provides a unified description of the topological and geometric properties of a manifold through the spectra of Hodge Laplacians. We then present the discretization of the theory on regular Cartesian grids, and introduce the Boundary-Induced Graph Laplacian used to generate the manifold embedding features in Section~\ref{Feature_Generation}. Finally, we describe its extension to the persistent setting by introducing persistent Hodge Laplacians that track the evolution of topological and geometric structures across a filtration of manifolds, together with their discretization for practical computation.

\subsection{De Rham--Hodge Theory}\label{sec.deRhamHodge}

De Rham--Hodge theory connects the topology of a manifold to the analysis of differential forms defined on it, and provides the mathematical foundation for the manifold embedding features used in this work. This subsection reviews the elements needed later: the differential and its cohomology, which are determined by topology alone, and the Hodge Laplacian, whose kernel recovers that topology while its nonzero spectrum reflects geometry.

Throughout, let $M$ be an $m$-dimensional smooth, orientable, compact Riemannian manifold with boundary $\partial M$, and let $\Omega^k(M)$ denote the space of differential $k$-forms on $M$. Differential forms are the objects that can be integrated over $k$-dimensional submanifolds, and this integrability is what ultimately makes them amenable to discretization.

The differential $d:\Omega^k(M)\to\Omega^{k+1}(M)$, also called the exterior derivative, is the unique $\mathbb{R}$-linear map that satisfies the Leibniz rule with respect to the wedge product and the nilpotent property
\begin{align}\label{eq.dd}
	dd = 0.
\end{align}
It is characterized by Stokes' theorem: for any oriented $(k\!+\!1)$-submanifold $S\subset M$ with boundary $\partial S$,
\begin{align}\label{eq.stokes}
	\int_S d\omega = \int_{\partial S}\omega,
\end{align}
which states that differentiating a form and integrating over a domain is the same as integrating the form over the boundary of that domain. Equation \eqref{eq.stokes} is more than a classical identity here: it is the property that transfers $d$ to the discrete setting, where it becomes the signed incidence matrix between cells (Section~\ref{laplacian_construction}).

The nilpotency \eqref{eq.dd} gives $d$ its topological content. A form $\omega\in\Omega^k(M)$ is closed if $d\omega = 0$ and exact if $\omega = d\zeta$ for some $\zeta\in\Omega^{k-1}(M)$; since $dd = 0$, every exact form is closed, but the converse may fail, and the extent of this failure measures the holes of $M$. Making this precise, $d$ assembles the spaces $\Omega^k(M)$ into the de Rham complex
\begin{align}
	0\longrightarrow\Omega^0(M)\overset{d}{\longrightarrow}\Omega^1(M)\overset{d}{\longrightarrow}\cdots\overset{d}{\longrightarrow}\Omega^m(M)\longrightarrow 0,
\end{align}
whose $k$-th cohomology is the de Rham cohomology group
\begin{align}
	H^k_{dR}(M) = \frac{\ker\left(d:\Omega^k(M)\to\Omega^{k+1}(M)\right)}{\operatorname{im}\left(d:\Omega^{k-1}(M)\to\Omega^k(M)\right)}.
\end{align}
By the de Rham theorem, $H^k_{dR}(M)$ is naturally isomorphic to the singular cohomology of $M$ with real coefficients. It is therefore a topological invariant: it depends only on the topology of $M$, and not on its metric.

So far no metric has been used, and $d$ alone yields no notion of adjointness and hence no Laplacian. Both require a way to measure forms, which the Riemannian metric supplies. Let $g$ be a metric on $M$ with pointwise inner product $\langle\cdot,\cdot\rangle_g$ and volume form $\mu_g$. The Hodge star $\star:\Omega^k(M)\to\Omega^{m-k}(M)$ is the isomorphism defined by
\begin{align}\label{eq.hodgestar}
	\omega\wedge\star\eta = \langle\omega,\eta\rangle_g\,\mu_g,
\end{align}
which pairs each $k$-form with a complementary $(m\!-\!k)$-form so that their wedge product is a top-degree form, and hence integrable. Integrating \eqref{eq.hodgestar} gives the Hodge $L^2$-inner product on $\Omega^k(M)$,
\begin{align}\label{eq.l2inner}
	(\omega,\eta) = \int_M\omega\wedge\star\eta.
\end{align}
With an inner product available, $d$ admits a formal adjoint. The codifferential $\delta:\Omega^k(M)\to\Omega^{k-1}(M)$ is defined by
\begin{align}\label{eq.codiff}
	\delta = (-1)^{m(k-1)+1}\star d\star,
\end{align}
and inherits the nilpotency $\delta\delta = 0$ from \eqref{eq.dd}. Whereas $d$ raises the degree of a form, $\delta$ lowers it; composing the two in both orders yields a degree-preserving operator on $\Omega^k(M)$, which is the central object of this subsection.

\begin{definition}[Hodge Laplacian]\label{def.hodgeLaplacian}
	The Hodge Laplacian on $\Omega^k(M)$ is the operator
	\begin{align}
		\Delta = d\delta + \delta d: \Omega^k(M)\to\Omega^k(M).
	\end{align}
	A form $\omega\in\Omega^k(M)$ is called harmonic if $\Delta\omega = 0$, and the space of harmonic $k$-forms is denoted by $\mathcal{H}^k_\Delta(M)$. A form is called coclosed if $\delta\omega = 0$, and the space $\mathcal{H}^k(M) = \ker d\cap\ker\delta$ of closed and coclosed $k$-forms is called the space of harmonic fields.
\end{definition}

\begin{remark}[Vector calculus in $\mathbb{R}^3$]\label{rem.vectorcalculus}
	Both $d$ and $\delta$ generalize the classical operators of vector calculus. In $\mathbb{R}^3$, $0$-forms and $3$-forms are identified with scalar fields, while $1$-forms and $2$-forms are identified with vector fields. Under these identifications, $d$ acts as the gradient $\nabla$ on $0$-forms, the curl $\nabla\times$ on $1$-forms, and the divergence $\nabla\cdot$ on $2$-forms, while $\delta$ corresponds to $-\nabla\cdot$, $\nabla\times$, and $-\nabla$ when applied to $1$-, $2$-, and $3$-forms, respectively. The nilpotent properties $dd = 0$ and $\delta\delta = 0$ then recover the familiar identities $\nabla\times\nabla f = 0$ and $\nabla\cdot(\nabla\times\mathbf{v}) = 0$.
\end{remark}

When $M$ is boundaryless, i.e., $\partial M = \emptyset$, the operators $d$ and $\delta$ are $L^2$-adjoint with respect to \eqref{eq.l2inner}, so that $(d\omega,\eta) = (\omega,\delta\eta)$. Consequently, whenever $\Delta\omega = 0$,
\begin{align}\label{eq.harmonicsplit}
	0 = (\Delta\omega,\omega) = (d\omega,d\omega) + (\delta\omega,\delta\omega),
\end{align}
so that harmonic forms are simultaneously closed and coclosed, and hence $\mathcal{H}^k_\Delta(M) = \mathcal{H}^k(M)$ in this case. The harmonic space then carries exactly the topological content of $M$.

\begin{theorem}[Hodge theorem~\cite{hodge1941theory}]\label{thm.hodge}
	Let $M$ be a closed, orientable Riemannian manifold. Then the space of harmonic $k$-forms is finite-dimensional and isomorphic to the $k$-th de Rham cohomology group,
	\begin{align}
		\mathcal{H}^k_\Delta(M)\cong H^k_{dR}(M),
	\end{align}
	so that $\dim\mathcal{H}^k_\Delta(M) = \beta_k$, the $k$-th Betti number of $M$.
\end{theorem}

\begin{proof}
	See~\cite{hodge1941theory}; a modern treatment is given in~\cite{schwarz2006hodge}.
\end{proof}

The Betti number $\beta_k$ counts the $k$-dimensional holes of $M$: $\beta_0$ is the number of connected components, $\beta_1$ is the number of tunnels or loops, and $\beta_2$ is the number of enclosed cavities. Theorem~\ref{thm.hodge} therefore identifies the kernel of the Hodge Laplacian with the topology of $M$, which is the property exploited throughout this work.

In the presence of a boundary, the adjointness of $d$ and $\delta$ fails, and \eqref{eq.harmonicsplit} no longer holds. Boundary conditions must therefore be imposed to ensure that the operators are well-posed.

\begin{definition}[Normal and tangential boundary conditions]\label{def.boundaryconditions}
	The spaces of normal and tangential $k$-forms on $M$ are defined, respectively, by
	\begin{align}
		\Omega^k_n(M) &= \{\omega\in\Omega^k(M)\,|\,\omega|_{\partial M} = \delta\omega|_{\partial M} = 0\},\\
		\Omega^k_t(M) &= \{\omega\in\Omega^k(M)\,|\,\star\omega|_{\partial M} = \star d\omega|_{\partial M} = 0\},
	\end{align}
	which are Hodge dual to each other. The corresponding spaces of \emph{normal} and \emph{tangential harmonic fields} are $\mathcal{H}^k_n(M) = \mathcal{H}^k(M)\cap\Omega^k_n(M)$ and $\mathcal{H}^k_t(M) = \mathcal{H}^k(M)\cap\Omega^k_t(M)$. We write $\Delta_n$ and $\Delta_t$ for the Hodge Laplacian restricted to $\Omega^k_n(M)$ and $\Omega^k_t(M)$, respectively.
\end{definition}

\begin{remark}\label{rem.harmonicfields}
	Unlike the boundaryless case, the space of harmonic fields $\mathcal{H}^k(M)$ is in general infinite-dimensional and is only a subset of $\ker\Delta$. It therefore admits no direct correspondence with the cohomology of $M$. Imposing the boundary conditions of Definition~\ref{def.boundaryconditions} restores this correspondence, as stated next.
\end{remark}

\begin{theorem}[Friedrichs~\cite{friedrichs1955differential}]\label{thm.friedrichs}
	Let $H^k_{dR}(M,\partial M)$ and $H^k_{dR}(M)$ denote the relative and absolute de Rham cohomology groups of $M$. Then
	\begin{align}
		\mathcal{H}^k_n(M) = \ker\Delta_n^k \cong H^k_{dR}(M,\partial M)
		\quad\text{and}\quad
		\mathcal{H}^k_t(M) = \ker\Delta_t^k \cong H^k_{dR}(M).
	\end{align}
	In particular, both spaces are finite-dimensional, with
	\begin{align}\label{eq.betti.nt}
		\dim\mathcal{H}^k_n(M) = \beta_{m-k}
		\quad\text{and}\quad
		\dim\mathcal{H}^k_t(M) = \beta_k.
	\end{align}
\end{theorem}

\begin{proof}
	See~\cite{friedrichs1955differential} and~\cite[Chapter 2]{schwarz2006hodge}.
\end{proof}

Theorem~\ref{thm.friedrichs} is the foundation of the present framework: it allows the topology of the underlying manifold to be recovered from the kernels of the boundary-constrained Hodge Laplacians, while their nonzero spectra provide additional insight into the geometric structure of $M$. For the compact domains $M\subset\mathbb{R}^3$ considered in this work, the case of interest is the following.

\begin{corollary}\label{cor.betti0}
	Let $M\subset\mathbb{R}^3$ be a compact domain, so that $m = 3$. Then the Hodge Laplacian on $3$-forms under the normal boundary condition satisfies
	\begin{align}
		\dim\ker\Delta_n^3 = \beta_{3-3} = \beta_0,
	\end{align}
	that is, the number of zero eigenvalues of $\Delta_n^3$ equals the number of connected components of $M$.
\end{corollary}

\begin{proof}
	Immediate from \eqref{eq.betti.nt} with $m = 3$ and $k = 3$.
\end{proof}

The discrete counterpart of $\Delta_n^3$, denoted $L_{3,n}$, is the operator used to construct the manifold embedding features of Section~\ref{PML_framework}; its discretization is presented next.

\subsection{Discretization and construction of Laplacians}\label{laplacian_construction}

The discretization of the de Rham--Hodge framework can be realized using discrete exterior calculus (DEC)~\cite{desbrun2006discrete}. This approach ensures consistency with the continuous theory and enables correct computation of the manifold cohomology. In addition, DEC provides straightforward treatment of the boundary conditions, and reduces the computation to matrix algebra with symmetric Hodge Laplacian matrices, thereby allowing for efficient numerical implementation. The DEC-based discretization framework has been developed in two different ways: the Eulerian formulation~\cite{su2024persistent,su2024topology}, in which the underlying manifold $M$ is represented as a sublevel set of a level-set function on the Cartesian grid, and the Lagrangian formulation~\cite{chen2021evolutionary}, which discretizes $M$ as a simplicial mesh. The Eulerian representation fixes all vertices, edges, faces, and cells to the grid, which substantially simplifies the resulting data structures and algorithms. It also eliminates the need to explicitly tessellate the domain into simplicial meshes, and a filtration of the manifold can be constructed simply by adjusting the isovalues. Due to these advantages, we adopt the Eulerian formulation developed in~\cite{su2024persistent}. Let $I_m$ be an $m$-dimensional Cartesian grid that contains the manifold $M$, where each $k$-cell is oriented according to its alignment with the coordinate axes. Below, we first present the construction of the discrete differential operators on the entire grid and then restrict them to $M$. Due to the Hodge duality between the normal and tangential boundary conditions, here we consider only the normal boundary conditions.

On the Cartesian grid $I_m$, a differential $k$-form $\omega$ can be discretized by integrating it over each oriented $k$-cell $\sigma_i$, i.e., $W^i = \int_{\sigma_i}\omega$. This is precisely the discrete counterpart of Stokes' theorem \eqref{eq.stokes}: the discrete exterior derivative $D^I_k$ is given by the signed incidence matrix between $(k\!+\!1)$-cells and $k$-cells, which is identified with the transpose of the cell boundary operator $\partial_{k+1}^T$. Since $\partial\partial = 0$ for the boundary of a boundary, this immediately gives $D^I_{k+1}D^I_k = 0$, mirroring the continuous identity $dd = 0$~\cite{desbrun2006discrete}. The discrete Hodge star establishes a correspondence between $k$-forms on the primal grid and $(m\!-\!k)$-forms on the dual grid (formed by translating grid points to cell centers), via
\begin{align}
	\frac{1}{|\sigma_k|}\int_{\sigma_k}\omega \approx \frac{1}{|\star\sigma_k|}\int_{\star\sigma_k}\star\omega,
\end{align}
which yields a diagonal matrix $S^I_k$ with diagonal entries equal to the ratio of dual $(m\!-\!k)$-cell to primal $k$-cell volumes, given by $\ell^{m-2k}$ for a grid of edge length $\ell$. The discrete codifferential is assembled from $D^I_k$ and $S^I_k$ as $\delta^I_k = (S^I_{k-1})^{-1}(D^I_{k-1})^TS^I_k$, and the discrete Hodge Laplacian, obtained as the counterpart of $\star\Delta$ to preserve symmetry, is given by
\begin{align}
	L^I_{k} = (D^I_{k})^TS^I_{k+1}D^I_{k} + S^I_{k}D^I_{k-1}(S^I_{k-1})^{-1}(D^I_{k-1})^TS^I_{k}.
\end{align}

For the discrete operators on $M$, note that the manifold boundary $\partial M$ generally intersects grid cells rather than aligning with them. The normal boundary condition is enforced by restricting the operators to the normal support of cells, where a $k$-cell belongs to the support if at least one of its vertices lies inside or on $\partial M$. Denote by $P_{k,n}$ the corresponding projection matrix onto the normal support. The differential and Hodge star operators for $M$ are then given by
\begin{align}
	D_{k,n} = P_{k+1,n}D^I_kP_{k,n}^T, \quad S_{k,n} = P_{k,n}S^I_{k,n}P_{k,n}^T,
\end{align}
where $S^I_{k,n}$ is obtained by adjusting the cell volumes in $S^I_k$ according to the normal boundary conditions following~\cite{su2024persistent}. Using these discrete operators, one can assemble the discrete Hodge Laplacian under the normal boundary conditions as follows:
\begin{align}\label{eq.discreteHodge}
	L_{k,n} = D_{k,n}^TS_{k+1,n}D_{k,n} + S_{k,n}D_{k-1,n}S_{k-1,n}^{-1}D_{k-1,n}^TS_{k,n},
\end{align}
whose null space, as in the continuous theory of Theorem~\ref{thm.friedrichs}, is fully determined by the manifold topology, with dimension given by the Betti number $\beta_{m-k}$. Beyond the kernel, the nonzero eigenvalues of the Laplacian further encode the geometric structure of $M$. For instance, the smallest nonzero eigenvalue (the Fiedler value) reflects the manifold's connectivity~\cite{fiedler1973algebraic}, while eigenvalue multiplicities can reveal underlying symmetries.

In the Eulerian setting, the discrete Hodge star is nearly a rescaled identity matrix, since $\ell^{m-2k}$ is constant on the interior of the grid and departs from this value only in the cells meeting the boundary. The Hodge Laplacian \eqref{eq.discreteHodge} can therefore be well approximated by replacing the Hodge star with the identity matrix, which leads to the following operator.

\begin{definition}[Boundary-Induced Graph Laplacian~\cite{ribando2024graph}]\label{def.BIGLaplacian}
	Let $D_{k,n}$ be the discrete differential on the normal support of $M$. The Boundary-Induced Graph (BIG) Laplacian on discrete $k$-forms under the normal boundary condition is
	\begin{align}\label{eq.BIGLaplacians}
		L^B_{k,n} = D_{k,n}^TD_{k,n} + D_{k-1,n}D_{k-1,n}^T,
	\end{align}
	that is, the discrete Hodge Laplacian \eqref{eq.discreteHodge} with every Hodge star matrix replaced by the identity.
\end{definition}

The spectra of the BIG Laplacians have been shown to converge to those of the corresponding Hodge Laplacians up to a scaling factor~\cite{ribando2024graph}. Therefore, they can similarly be used to capture the topological and geometric information of $M$ through their spectra as the Hodge Laplacians, while being computationally cheaper, since their construction does not require Hodge star matrices.

The spectral analysis of the Hodge and BIG Laplacians can be treated uniformly by writing a generic Laplacian $L_k = D_k^TS_{k+1}D_k + S_kD_{k-1}S_{k-1}^{-1}D_{k-1}^TS_k$, which reduces to the Hodge or BIG Laplacian depending on whether $S_k$ is the true Hodge star or the identity matrix. The spectrum of $L_k$ can be obtained by considering the generalized eigenvalue problem $L_kW = \lambda S_kW$, which can also be reformulated as the ordinary eigenvalue problem $\bar L_k\bar W = \lambda\bar W$ through the change of variables
\begin{align}
	\bar D_k = S_{k+1}^{1/2}D_kS_k^{-1/2}, \qquad \bar L_k = S_k^{-1/2}L_kS_k^{-1/2}, \qquad \bar W = S_k^{1/2}W.
\end{align}
Here the transformed Laplacian is of the form
\begin{align}
	\bar L_k = \bar D_k^T\bar D_k + \bar D_{k-1}\bar D_{k-1}^T.
\end{align}
Let $\bar D_k = U_{k+1}\Sigma_kV_k^T$ be the singular value decomposition of $\bar D_k$. As shown in~\cite{su2024persistent}, the spectrum of $\bar L_k$ is given by the union of the squared nonzero singular values of $\bar D_k$ and $\bar D_{k-1}$, together with the zero eigenvalue whose multiplicity is the $(m\!-\!k)$-th Betti number. Moreover, the space of discrete $k$-forms is spanned by the harmonic forms together with the columns of $U_k$ and $V_k$ associated with the nonzero singular values.

For compact domains in $\mathbb{R}^3$, there are four Laplacians $L_{k,n}$, $k = 0,1,2,3$, of different degrees under the normal boundary conditions. The spectral analysis of these Laplacians can therefore be reduced to the analysis of the singular value spectra of the three discrete differentials $\bar D_0$, $\bar D_1$, and $\bar D_2$.

\begin{remark}\label{rem.L3n}
	Among these four operators, the present work uses only $L^B_{3,n}$, the BIG Laplacian on discrete $3$-forms under the normal boundary condition, which we denote simply by $L_{3,n}$. By Corollary~\ref{cor.betti0}, the number of its zero eigenvalues equals the $0$-th Betti number $\beta_0$, so that its spectrum supplies $\beta_0$ together with a sequence of nonzero eigenvalues reflecting the geometry of $M$. Higher-order Laplacians could provide additional descriptors, but are not pursued here for reasons of computational efficiency.
\end{remark}

\subsection{Persistent Hodge Laplacian}\label{sec.persistentHodge}

This section presents the construction of the persistent Hodge Laplacian on a filtration of manifolds obtained by varying a single filtration parameter, which was first introduced in~\cite{chen2021evolutionary}, with the discretization framework for the Eulerian formulation presented later in~\cite{su2024persistent}. Compared with the Hodge Laplacian at a single scale, the persistent Hodge Laplacian captures how the manifold topology and geometry change across multiple scales, and therefore provides substantially more information than that of any single manifold.

In practice, a filtration of manifolds can be constructed directly through a level set function. To be specific, let $f:N\to[a,b]$ be a function defined on the ambient space $N$ with compact sublevel sets under consideration. We assume $f$ is a Morse function, as the set of Morse functions is dense in the space of smooth functions on a compact manifold. We then choose evenly distributed isovalues $a\leq c_0<c_1<\cdots<c_s\leq b$ and perturb them slightly, if necessary, so that none of these isovalues are critical values of $f$. This is always possible, as the critical values are isolated and finite for compact manifolds. Each sublevel set $M_i = \{x\in N\,|\,f(x)\leq c_i\}$, $i=0,\ldots,s$, is therefore a smooth manifold with boundary $\partial M_i = \{x\in N\,|\,f(x)=c_i\}$. These manifolds form a filtration linked by the inclusion maps
\begin{align}\label{eq.filtration}
	M_0\hookrightarrow M_1\hookrightarrow M_2\hookrightarrow\cdots\hookrightarrow M_{s-1}\hookrightarrow M_s.
\end{align}
In addition, we assume that $M_{i,i+1} = \overline{M_{i+1}\backslash M_i} = \{x\in N\,|\,f(x)\in[c_i,c_{i+1}]\}$ contains at most one critical point for each $i=0,\ldots,s$, which can always be realized by refining the sequence of isovalues. In the case that $M_{i,i+1}$ contains no critical points, the manifold $M_{i+1}$ retracts to $M_i$ along the gradient flow of $f$, and thus they are homotopy equivalent, meaning that no topological changes occur between $c_i$ and $c_{i+1}$. If $M_{i,i+1}$ contains one critical point, then the topological change happens precisely at the critical value of $f$. Depending on the type of the critical point (i.e., local minimum, saddle point, or local maximum), the topology changes in different ways~\cite{su2024persistent}.

Now let $M_i\hookrightarrow M_{i+1}$ be a consecutive pair of manifolds in the filtration. A filtration-induced extension from $\Omega^k_n(M_i)$ to $\Omega^k_n(M_{i+1})$ that extends each normal $k$-form on $M_i$ to a normal $k$-form on $M_{i+1}$ can be realized by solving a biharmonic equation $\Delta^2\zeta = \Delta(\Delta\zeta) = 0$ on $M_{i,i+1}$ with both Dirichlet and Neumann boundary conditions to ensure that $d\zeta$ joins smoothly with $\omega$ through $\partial M_i$. The extended form $\bar\omega$ on $M_{i+1}$ can then be defined so that $\bar\omega = \omega$ on $M_i$ and $\bar\omega = d\zeta$ on $M_{i,i+1}$. Note that $d\zeta$ satisfies the normal boundary condition on $\partial M_{i+1}$, so that $\bar\omega$ is normal, i.e., $\bar\omega\in\Omega^k_n(M_{i+1})$. Denote by $\mathcal{I}^k_{i,i+1}:\Omega^k_n(M_i)\to\Omega^k_n(M_{i+1})$ this extension. The following commutative diagram then follows from the uniqueness of the solution to the boundary value problem
\begin{equation}
\label{eq:comDiagm1}
\begin{tikzcd}
	\Omega^{0}_n(M_0) \arrow{r}{d} \arrow{d}{\mathcal{I}^0_{0,1}} & \Omega^1_n(M_0) \arrow{r}{d} \arrow{d}{\mathcal{I}^1_{0,1}} & \Omega^{2}_n(M_0) \arrow{r}{d} \arrow{d}{\mathcal{I}^2_{0,1}}  & \Omega^{3}_n(M_0) \arrow{d}{\mathcal{I}^3_{0,1}} \\
	\Omega^{0}_n(M_1) \arrow{r}{d} \arrow{d}{\mathcal{I}^0_{1,2}} & \Omega^1_n(M_1)  \arrow{r}{d} \arrow{d}{\mathcal{I}^1_{1,2}} & \Omega^{2}_n(M_1)  \arrow{r}{d} \arrow{d}{\mathcal{I}^2_{1,2}}  & \Omega^{3}_n(M_1)  \arrow{d}{\mathcal{I}^3_{1,2}} \\
	\Omega^{0}_n(M_2)  \arrow{r}{d} \arrow{d}{\mathcal{I}^0_{2,3}} & \Omega^1_n(M_2) \arrow{r}{d} \arrow{d}{\mathcal{I}^1_{2,3}} & \Omega^{2}_n(M_2) \arrow{r}{d} \arrow{d}{\mathcal{I}^2_{2,3}}  & \Omega^{3}_n(M_2) \arrow{d}{\mathcal{I}^3_{2,3}} \\
	\cdots  & \cdots  & \cdots  & \cdots
\end{tikzcd}
\end{equation}
where each row represents the de Rham complex for each manifold in the filtration, and each column represents the extensions induced by the filtration.

For a pair of manifolds $M_i\hookrightarrow M_j$ with $i\leq j$, we consider the following commutative diagram
\begin{equation}
\label{eq:comDiagm2}
\begin{tikzcd}
	{ } && {\Omega_n^{k}(M_{i})} && {\Omega_n^{k+1}(M_{i})} \\
	\\
	{\Omega_n^{k-1}(M_{j})} && {\Omega_n^{k}(M_{j})} && { }
	\arrow["{d_{j}^{k-1}}", shift left, from=3-1, to=3-3]
	\arrow["{\mathcal{R}_{i,j}}", shift left, from=3-3, to=1-3]
	\arrow["{\delta_{j}^k}", shift left, from=3-3, to=3-1]
	\arrow["{d_{i}^k}", shift left, from=1-3, to=1-5]
	\arrow["{{\delta}_{i,j}^k}", shift left=-2, from=1-3, to=3-1]
	\arrow["{\mathcal{I}_{i,j}}", shift left, from=1-3, to=3-3]
	\arrow["{\delta_{i}^{k+1}}", shift left, from=1-5, to=1-3]
	\arrow["{{d}_{i,j}^{k-1}}", shift left=4, from=3-1, to=1-3]
\end{tikzcd}.
\end{equation}
Here $\mathcal{I}_{i,j} = \mathcal{I}_{j-1,j}\circ\cdots\circ\mathcal{I}_{i,i+1}$ denotes the extension map from $\Omega^k_n(M_i)$ to $\Omega^k_n(M_j)$; $d_i$ and $\delta_i$ (resp. $d_j$ and $\delta_j$) denote the differential and codifferential on $\Omega^k(M_i)$ (resp. $\Omega^k(M_j)$); $\mathcal{R}_{i,j}$ is the projection from $\Omega^k_n(M_j)$ onto the space spanned by the extensions followed by the restriction to $M_i$; and
\begin{align}\label{eq.persistentOperators}
	{\delta}_{i,j}^k = \delta^k_j\circ\mathcal{I}_{i,j},
	\qquad
	{d}_{i,j}^{k-1} = \mathcal{R}_{i,j}\circ d^{k-1}_j.
\end{align}
These two operators are adjoint with respect to the Hodge inner product \eqref{eq.l2inner}, which allows the following operator to be defined consistently with the single-scale case.

\begin{definition}[Persistent Hodge Laplacian]\label{def.persistentHodge}
	Let $M_i\hookrightarrow M_j$ be a pair of manifolds in the filtration \eqref{eq.filtration} with $i\leq j$, and let $d^{k-1}_{i,j}$ and $\delta^k_{i,j}$ be the operators defined in \eqref{eq.persistentOperators}. The \emph{$k$-th persistent Hodge Laplacian} $\Delta^k_{i,j}:\Omega^k_n(M_i)\to\Omega^k_n(M_i)$ under the normal boundary condition is
	\begin{align}\label{eq.persistentHodgeLaplacian}
		\Delta^k_{i,j} = d^{k-1}_{i,j}\delta^k_{i,j} + \delta^{k+1}_i d^k_i.
	\end{align}
	Its kernel $\mathcal{H}^k_{i,j} = \ker\Delta^k_{i,j}$ is called the space of \emph{persistent normal harmonic fields}.
\end{definition}

\begin{remark}\label{rem.persistentReduction}
	The persistent Hodge Laplacian generalizes the operator of Section~\ref{sec.deRhamHodge} rather than replacing it. When $i = j$, the extension $\mathcal{I}_{i,i}$ is the identity, so that $d^{k-1}_{i,i} = d^{k-1}_i$ and $\delta^k_{i,i} = \delta^k_i$, and \eqref{eq.persistentHodgeLaplacian} reduces exactly to the Hodge Laplacian $\Delta^k_n$ on $M_i$ under the normal boundary condition. The single-scale theory is thus recovered as the special case in which no evolution of the manifold is considered.
\end{remark}

\begin{theorem}[\cite{chen2021evolutionary}]\label{thm.persistentBetti}
	The space of persistent normal harmonic fields admits the identification
	\begin{align}
		\mathcal{H}^k_{i,j} = \ker\delta^k_{i,j}\cap\ker d^k_i,
	\end{align}
	and is isomorphic to the $k$-th persistent de Rham cohomology group of the pair $M_i\hookrightarrow M_j$. In particular, $\mathcal{H}^k_{i,j}$ is finite-dimensional, with dimension given by the $(m\!-\!k)$-th persistent Betti number.
\end{theorem}

\begin{proof}
	See~\cite{chen2021evolutionary}; the corresponding discrete statement in the Eulerian representation is established in~\cite{su2024persistent}.
\end{proof}

\begin{remark}\label{rem.persistentUse}
	Theorem~\ref{thm.persistentBetti} extends the correspondence of Theorem~\ref{thm.friedrichs} to the multiscale setting: the kernel of $\Delta^k_{i,j}$ records the topological features of $M_i$ that persist into $M_j$, while its nonzero spectrum reflects the accompanying geometric variation. The manifold embedding features of Section~\ref{PML_framework} are obtained from the special case $i = j$ of Remark~\ref{rem.persistentReduction}, evaluated at each isovalue of the filtration \eqref{eq.filtration} and assembled across scales. Features drawn from the genuinely persistent case $i < j$ are not pursued in this work, but constitute a natural extension of the present framework.
\end{remark}

\subsection{Discretization of the persistent Hodge Laplacian}

Let $f:\mathbb{R}^m\to [a, b]$ be a levelset function with compact sublevel sets under consideration. By varying the isovalues of $f$, a nested sequence of sub-cell complexes under the normal boundary conditions can be obtained 
\begin{align}
	\emptyset = K_0\subset K_1\subset\cdots\subset K_{s-1}\subset K_s = K.
\end{align}
In the discrete setting, a discrete $k$-form on $K_i$ can be easily extended to $K_{i+1}$ by solving the discrete Laplace equation with certain boundary conditions on the closure $K_{i, i+1} = \operatorname{Cl}(K_{i+1}\backslash K_i)$. Denote by $I_{i, i+1}: C^k(K_i)\to C^k(K_{i+1})$ this extension. Then the extension mapping from $C^k(K_i)$ to $C^k(K_j)$ is given by $I_{i, j} = I_{j-1, j}\circ\cdots\circ I_{i, i+1}$. Alternatively, this extension $\mathcal{I}_{i,j}$ can be constructed directly by solving the Laplace equation on the closure $K_{i, j} = \operatorname{Cl}(K_j\backslash K_i)$, and we adopt this approach in the discretization below for simplicity. Through this extension, the space of discrete forms on $K_i$ can be identified with a subspace of the space of discrete forms on $K_{j}$. The discrete counterparts of the commutative diagram \eqref{eq:comDiagm1} and~\eqref{eq:comDiagm2} can then be obtained. 

Denote by $D_i^k$ and $\bar\delta_i^k$ the discrete differential and codifferential operator on $K_i$, respectively. Let $I^k_{i,j}$ and $R^k_{i,j}$ be the discrete operators on discrete $k$-forms corresponding to the extension mapping $\mathcal{I}_{i,j}$ and the projection $\mathcal{R}_{i,j}$, and let $D_{i,j}^{k-1}=R^k_{i,j}D_j^{k-1}$ and $\bar\delta_{i,j}^k = \bar\delta_{j}^kI^k_{i,j}$ the discrete operators corresponding to $d_{i,j}^{k-1}$ and $\delta_{i,j}^k$ in the previous section. 
As aforementioned, given a discrete $k$-form $W$ on $K_i$, the discrete extension $\mathcal{I}_{i,j}$ can be constructed as follows: let $\tilde L^{k-1}_{K_{i,j}}$ be the Laplace operator acting on $(k\!-\!1)$-forms on $K_{i,j}\backslash\partial K_i$ and let $\delta^k_{\partial K_i}$ denote the boundary codifferential operator that uses values of $W$ on $\partial K_i$ to evaluate the neighboring $(k\!-\!1)$-cells in $K_{i,j}\backslash {\partial K_i}$. Then solving the linear system $\tilde L^{k-1}_{K_{i,j}}\tilde\zeta = -S^{k-1}_{K_{i,j}}\delta^k_{\partial K_i}W$ gives rise to the required $d\tilde\zeta$ on $K_{i,j}\backslash\partial K_i$. Together with $W$ on $K_i$, this defines an extension in $C^k(K_j)$. It follows that the extension operator is given by:
\[I^k_{i,j}=
\begin{pmatrix}
\operatorname{Id}^k_{K_i}\\
-D^{k-1}_{K_{i,j}}(\tilde{L}^{k-1}_{K_{i,j}})^{-1} S^{k-1}_{K_{i,j}} \delta^k_{\partial K_i}
\end{pmatrix},
\]
where $\operatorname{Id}^k_{K_i}$ is the identity matrix on $K_i$ up to a rescaling on the boundary.

By definition, $\bar\delta_{i,j}^k = \bar\delta_{j}^kI^k_{i,j}$ and $\bar\delta_{j}^k = (S^{k-1}_j)^{-1}(D^{k-1}_j)^TS^k_j$. The operator $D_{i,j}^{k-1}$ is then of the form
\begin{align}
    D_{i,j}^{k-1} = (S_i^k)^{-1}(I^k_{i,j})^TS^k_jD^{k-1}_j,
\end{align}
following from its adjointness with $\bar\delta_{i,j}^k$. The projection $R^k_{i,j} = (S_i^k)^{-1}(I^k_{i,j})^TS^k_j$ can be viewed as the $L^2$-projection onto the subspace of forms on $K_j$ formed by all extensions from $K_i$, represented as forms on $K_i$. 

The discrete $k$-th persistent Hodge Laplacian on $k$-forms on $K_i$ can then be assembled as 
\begin{align}
	L_{i,j}^k = D_{i,j}^{k-1}\bar\delta_{i,j}^k + \bar\delta_{i}^{k+1}D_i^k,
\end{align}
and the discrete $k$-th persistent BIG Laplacian is
\begin{align}
	L_{i,j}^{B,k} = D_{i,j}^{k-1}(D_{i,j}^{k-1})^T + (D_i^k)^TD_i^k.
\end{align}
The discrete persistent BIG Laplacian $L_{i,j}^{B,k}$ has the same algebraic structure as the discrete persistent Hodge Laplacian $L_{i,j}^k$ but avoids the Hodge star matrices, which simplifies the computations.

\section{Persistent Manifold Learning Framework}
\label{PML_framework}

This section presents the proposed framework for BA prediction in metalloprotein-ligand and protein-protein complexes, organized as in Figure~\ref{fig:workflow}. Starting from a PDB structure in the SKEMPI v2~\cite{jankauskaite2019skempi} or metalloprotein-ligand~\cite{jiang2023metalprognet} dataset, element-specific atomic point clouds are extracted (panel a) and converted into a family of manifolds, generated as level sets of an atomic density function and discretized on a Cartesian grid into cell complexes (panel b). BIG Laplacians are then assembled on these complexes, and their spectra-harmonic (Betti numbers) and nonharmonic-provide the manifold embedding features that encode the topology and geometry of the binding interface across the filtration (panel c). In parallel, sequence-level embeddings from ESM-2~\cite{lin2023evolutionary,rives2021biological} for proteins and ChemBERTa~\cite{chithrananda2020chemberta} for ligands supply complementary information (panel d). The two descriptor sets are concatenated and used to train a gradient boosting decision tree for BA prediction (panel e).

\begin{figure}[!htbp]
    \centering
    \includegraphics[width=1\linewidth]{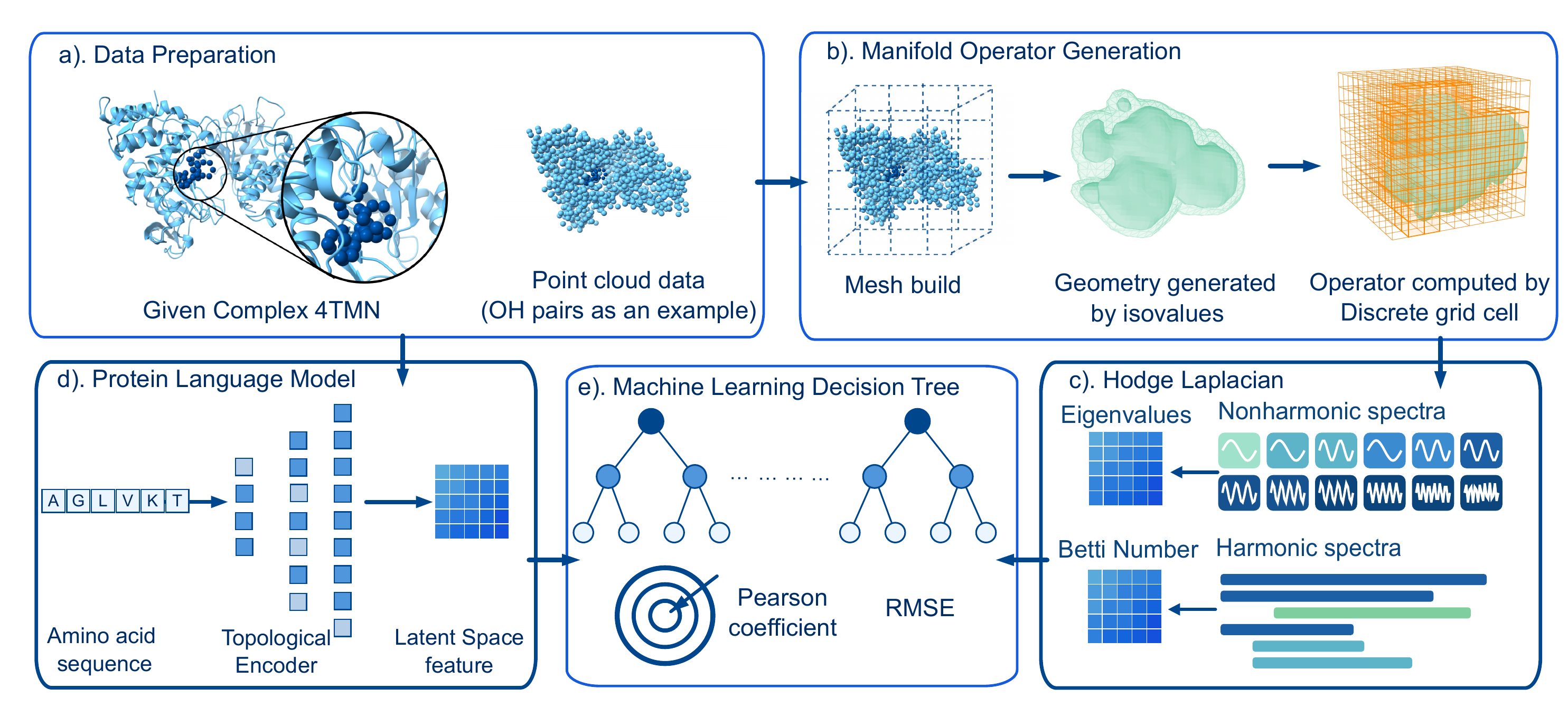}
    \caption{Overview of the proposed PML framework for metalloprotein-ligand and protein-protein BA prediction. The workflow includes manifold generation from protein structures in the SKEMPI v2 and metalloprotein--ligand datasets, geometry-based algebraic topology feature extraction using multiscale Hodge Laplacian descriptors, and machine learning model training and testing for BA prediction.}
    \label{fig:workflow}
\end{figure}
\subsection{Data Source and Pre-processing}

We evaluate the proposed framework on two benchmark datasets, one for each interaction type. For PPI binding affinity (BA) prediction, we adopt the wild-type subset of SKEMPI v2~\cite{jankauskaite2019skempi}, referred to as SKEMPI-WT. Wild-type complexes, rather than their mutant counterparts, are chosen because mutations typically induce only minor local geometric perturbations, making wild-type structures a more suitable and fundamental testbed for evaluating geometry-based topological representations. Since PDB structures in this dataset frequently contain missing atoms or residues, we first repair each structure using Profix from the Jackal package~\cite{xiang2001extending}. After this preprocessing step, the resulting SKEMPI-WT dataset comprises 343 protein-protein complexes.

For metalloprotein-ligand BA prediction, we adopt the metalloprotein-ligand dataset from~\cite{jiang2023metalprognet}, which provides protein-ligand complexes with experimentally measured binding affinities across a range of metal-ion-containing binding environments. In contrast to the PPI structures, these complexes are structurally complete and require no repair prior to use. We follow the original training-test partition from~\cite{wang2025join}, comprising 1,845 training complexes and 618 test complexes, for a total of 2,463 complexes. Model selection and hyperparameter tuning are carried out via 10-fold cross-validation on the training set, and final performance is reported on the independent test set.

\subsection{Element-Specific Manifold Construction}\label{Manifold_construction}

Each complex, whether a protein-protein complex or a metalloprotein-ligand complex, is represented by a set of point cloud data. To convert this discrete representation into the continuous manifolds required for our framework, we adopt the flexibility-rigidity index (FRI)~\cite{nguyen2019dg} as a discrete-to-continuum mapping. Following the element-specific approach of~\cite{cang2017topologynet}, we group atoms by element type and consider pairwise interactions between atom types on the two sides of the interface separately, rather than treating all atoms uniformly. For metalloprotein-ligand complexes, the protein side consists of Hydrogen (H), Carbon (C), Nitrogen (N), Oxygen (O), and Sulfur (S), while the ligand side additionally includes Phosphorus (P), Fluorine (F), Chlorine (Cl), Bromine (Br), and Iodine (I), reflecting the broader chemical diversity of small-molecule ligands; this yields 50 possible atom-type pairs. Since most proteins contain negligible Hydrogen content in the crystal structures used here, we exclude H from the protein side, reducing the number of atom pairs used in practice to 40, formed between protein atom types\(\{\text{C}, \text{N}, \text{O}, \text{S}\}\)
and ligand atom types ~\(\{\text{H}, \text{C}, \text{N}, \text{O}, \text{S}, \text{P}, \text{F}, \text{Cl}, \text{Br}, \text{I}\}\). In contrast, both sides of a protein-protein interface are ordinary protein chains, lacking the heteroatoms characteristic of small-molecule ligands; we therefore restrict both sides to C, N, and O, again excluding Sulfur and Hydrogen for the same reasons as above, yielding 
9 atom-type pairs between \(\{\text{C}, \text{N}, \text{O},\text{S}\}\) and \(\{\text{C}, \text{N}, \text{O},\text{S}\}\). For each complex, we restrict our attention to atoms within a cutoff distance of \(12\mathring{\text{A}}\) from the interface and use the resulting atom pairs to construct manifolds for each atom-pair type.

For a given atom-pair type, let \(\{\mathbf{x}^{\alpha}_i\}_{i=1}^{s}\)
 denote the positions of the \(s\) atoms of that type within the cutoff distance, where \(\alpha\) indicates which side of the interface the atom belongs to. We then define a level-set function on \(\mathbb{R}^3\) as the negative sum of Gaussian densities centered at these atomic positions:
\begin{align}\label{eq.lvf.rho}
\rho(\mathbf{x}, \tau) = -\sum^{s}_{i = 1}\exp\left(-\left(\frac{||\mathbf{x} - \mathbf{x^{\alpha}_i}||}{\tau r^{\alpha}_i}\right)^2\right),
\end{align}
where \(||\mathbf{x} - \mathbf{x^{\alpha}_i}||\)
is the Euclidean distance from \(\mathbf{x}\) to the \(i\)-th atom \(\mathbf{x^{\alpha}_i}\); \(r^{\alpha}_i\)
is its van der Waals radius, and \(\tau\) is a scale parameter. For a given isovalue \(c\), the sublevel set
\begin{align}
M = \{\mathbf{x}\,|\, \rho(\mathbf{x}, \tau)\leq c \}
\end{align}
defines a compact manifold in \(\mathbb{R}^3\), with boundary given by the isosurface \(\partial M = \{\mathbf{x} \, | \, \rho(\mathbf{x}, \tau) = c \}\). Choosing a sequence of increasing isovalues \(c_1 < c_2 < \cdots < c_s\) yields a nested filtration of manifolds,
\begin{align}
M_1\subset M_2\subset \cdots \subset M_s,
\end{align}
where \(M_i\) is the sublevel-set manifold associated with isovalue \(c_i\). In our PML experiments, we use \(s=9\) evenly spaced isovalues in the interval \([-0.5, -0.001]\)
to generate this filtration for each atom pair; this range is chosen because isovalues above \(-0.001\) produce almost no change in the topology of the resulting manifolds for most atom pairs, while isovalues below \(-0.5\) require finer grids to resolve the corresponding isosurfaces, substantially increasing computational costs. Figure~\ref{fig:4tmn_oh} illustrates this filtration at three representative isovalues for the O-H atom pair in the complex 4TMN as an example. We note that Eq.~\eqref{eq.lvf.rho} is one particular choice of FRI density function, which has been shown to convert discrete point-cloud data into continuous embeddings in a numerically stable manner, and has previously been used to construct protein boundary surfaces~\cite{chen2021evolutionary} and interactive manifolds~\cite{nguyen2019dg}; other reasonable choices of FRI density functions could equally be used to generate the filtration.

\begin{figure}[htbp]
\centering
\includegraphics[width=1\linewidth]{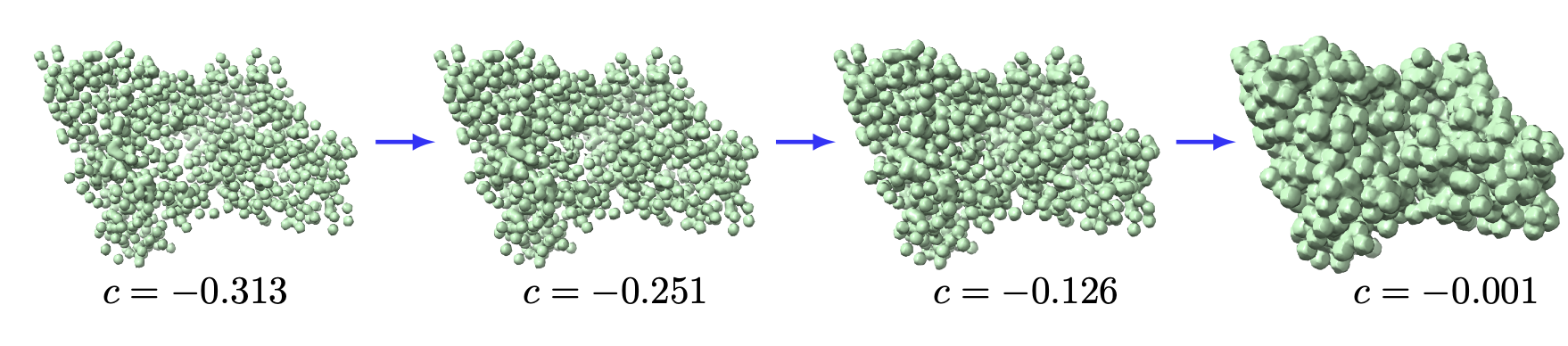}
\caption{
Homotopic evolution of the Gaussian isosurface manifold generated from
protein oxygen atoms and ligand hydrogen atoms in the 4TMN protein--ligand
complex. The manifold is constructed from a Gaussian density field over
OH atom pairs and evolves as the isovalue decreases
}

\label{fig:4tmn_oh}

\end{figure}

\subsection{Grid Construction and Cell Complex}

To compute the manifolds described in Section~\ref{Manifold_construction} numerically, we embed each atom pair in a regular Cartesian grid. For a given atom pair, a bounding box is first constructed by taking the minimum and maximum coordinates over all atoms of that type within the cutoff distance along each axis, and then expanded by 25\% of its range on each side of every axis, so that the resulting manifold does not touch the boundary of the grid. This bounding box is then discretized into a uniform $N\times N\times N$ lattice of grid vertices ($N=100$ in our experiments), with equally spaced coordinates along each axis. We evaluate the level-set function $\rho(\mathbf{x},\tau)$ from Eq.~(2.1) at every grid vertex, producing a discretized approximation of the density field used to construct the manifolds in the filtration.

This vertex lattice induces a structured cell complex consisting of $N^3$ vertices, $3(N-1)N^2$ edges, $3(N-1)^2N$ faces, and $(N-1)^3$ volumetric cells, providing the combinatorial structure required by the discrete exterior calculus framework introduced in Section~\ref{Methods}. For a given isovalue $c_i$, we include a grid cell in the truncated complex representing $M_i$ if at least one of its vertices lies in the corresponding sublevel set defined in Section~\ref{Manifold_construction}; cells with no such vertex are discarded. This yields, for each isovalue in the filtration, a truncated cell complex conforming to the corresponding manifold $M_i$, which provides the discrete geometric structure needed to compute the manifold embedding features introduced as follows.

\subsection{Features Generation}\label{Feature_Generation}
Manifold embedding features characterize a binding interface by representing the atoms involved as samples from an underlying manifold, and then extracting spectral and topological information from BIG Laplacians defined on that manifold, denoted \(L_{3,n}\)
and defined in detail in Section~\ref{laplacian_construction}. This representation applies broadly to different types of binding interfaces, including metalloprotein-ligand binding sites and protein-protein interfaces, though the specific construction of the manifold depends on the system under study. For metalloprotein-ligand complexes, we derive these manifolds for all complexes from atomic density fields defined around each pair of interacting atom types and focus on the 0-th order BIG Laplacian, whose zero eigenvalues give the 0-th Betti number \(\beta_0\) and whose nonzero eigenvalues capture the underlying geometric structure of the binding interface; higher-order Laplacians could add further features, but we restrict to \(L_{3,n}\) here for computational efficiency.

For each atom pair, we use a fixed Cartesian grid that remains the same across all complexes, which keeps spectra comparable while avoiding the high computational cost of a single grid that is large and fine enough to cover all complexes. On each manifold, we compute \(L_{3,n}\)
under the normal boundary condition and extract \(\beta_0\) together with the first 2 non-zero eigenvalues as the manifold embedding feature, respectively, giving \(6\) features per manifold in total. Across 9 isovalues and 40 atom-pair types, this yields \(3\times9\times40\) manifold embedding features per protein-ligand complex.
All manifold constructions and Laplacian computations are performed on the High Performance Computing Center, using B200 GPU nodes to accelerate the large-scale grid and eigenvalue computations required across all atom pairs and complexes. This feature set is sufficient to validate the proposed framework for predicting binding affinities, as shown in Section~\ref{results}.

\subsection{Protein Language Models}
Beyond geometry and topology, sequence-derived information offers an additional, complementary view of a binding complex. To capture this, embeddings from pretrained protein and molecular language models are incorporated alongside the manifold embedding features. Two such models are used in this work: ESM-2~\cite{lin2023evolutionary,rives2021biological}, a transformer pretrained on protein sequences, and ChemBERTa~\cite{chithrananda2020chemberta}, a transformer pretrained on molecular SMILES strings. Which of the two is applied depends on whether a given partner in the complex is a protein or a small-molecule ligand.

In the PPI setting, both partners are proteins, so both are encoded with ESM-2, a 650-million-parameter, 33-layer model trained on roughly 250 million sequences under a masked-token objective. Each protein chain is passed through the model independently, and its residue-level embeddings from the final layer are averaged over all residues to obtain a single 1280-dimensional vector. Concatenating the vectors for the two partners gives a 2560-dimensional sequence embedding for each PPI complex.
The MPLI setting instead pairs a protein with a small-molecule ligand, so a second model is needed to handle the ligand side. The protein component is embedded with the same ESM-2 model described above, again yielding a 1280-dimensional vector. The ligand, represented as a SMILES string, is embedded separately using ChemBERTa-77M-MLM, a RoBERTa-style model pretrained on 77 million SMILES strings via masked language modeling; its final hidden states are averaged across tokens to produce a 384-dimensional vector. The protein and ligand embeddings are concatenated into a single 1664-dimensional sequence embedding for each MPLI complex.
Since ESM-2 has a fixed maximum input length, protein sequences longer than this limit are split into subsequences, each embedded separately before the resulting per-residue embeddings are pooled together so that sequence embeddings remain available even for long protein chains.

\subsection{Machine Learning Platform}
Binding affinity prediction is carried out using the Gradient Boosting Decision Tree (GBDT) algorithm, implemented via the Python library. GBDT is a widely used ensemble method known for its resistance to overfitting, stable performance across a range of hyperparameter choices, and straightforward implementation. Rather than relying on a single predictor, the method builds a sequence of shallow decision trees, each one trained to correct the errors left by the trees before it; while any individual tree is a weak predictor on its own, aggregating many such trees yields a much stronger overall model. The manifold embedding features and protein language model embeddings introduced in Section~2 serve as the input to this GBDT model for BA regression. Table~\ref{tab:gbdt-hyperparams} summarizes the hyperparameter configuration used throughout. Because GBDT training involves stochastic subsampling, results can vary across runs; to obtain a stable estimate of performance, each model is retrained from 10 different random seeds (42, 1234, 5678, 91011, 121314, 151617, 181920, 212223, 242526, and 272829), and the final reported metrics are averaged across these 10 independent runs.
\begin{table}[!htbp]
\centering
\caption{Hyperparameters used to build the gradient boosting regression models.}
\label{tab:gbdt-hyperparams}
\begin{tabular}{ccc}
\toprule
Number of Estimators & Max depth & Minimum sample split \\
\midrule
10000 & 5 & 5 \\
\midrule
Loss function & Subsample & Max features \\
\midrule
Squared error & 0.5 & Square root \\
\bottomrule
\end{tabular}
\end{table}

\section{Results}\label{results}

To quantitatively evaluate the performance of our binding affinity prediction models, we employ Pearson's correlation coefficient (PCC), and root mean squared error (RMSE). PCC measures the strength and direction of the linear relationship between predicted and experimental binding affinities, making it well suited for assessing how well a model preserves the relative ranking of binding affinities across complexes; RMSE, in contrast, penalizes large errors more heavily due to the squared term. Reporting both types of metrics together thus provides a more complete assessment of model performance than either alone, capturing correlation and absolute accuracy simultaneously.
PCC is defined as
\begin{align}
\text{PCC}(\mathbf{y}, \mathbf{y}^p) = \frac{\sum_{i=1}^{n}(y_i - \bar{y})(y_i^p - \bar{y}^p)}{\sqrt{\sum_{i=1}^{n}(y_i - \bar{y})^2}\sqrt{\sum_{i=1}^{n}(y_i^p - \bar{y}^p)^2}},
\end{align}
where \(y_i\) and \(y_i^p\) denote the experimental and predicted binding affinity values for the \(i\)-th sample, respectively, and \(\bar{y}\)
and \(\bar{y}^p\) are their corresponding mean values.
RMSE are computed as
\begin{align}
\text{RMSE} = \sqrt{\frac{1}{n}\sum_{i=1}^{n}(y_i - y_i^p)^2},
\end{align}
where \(n\) is the total number of samples.
We employ these metrics to assess the performance of our machine learning models on both datasets. To ensure an unbiased evaluation, PCC and RMSE are computed via 10-fold cross-validation rather than on a single train-test split, so that the reported performance reflects the model's ability to generalize across different subsets of the data rather than its fit to one particular partition.
\subsection{Performance of Metallprotein-ligand binding affinity predictions} 
\subsubsection{Comparison with Existing Methods}
We first evaluate the proposed PML framework on the metalloprotein-ligand dataset from~\cite{jiang2023metalprognet} and compare its performance against four previously published models: RosENet~\cite{hassanharrirou2020rosenet}, NNScore2.0~\cite{durrant2011nnscore}, MetalProGNet~\cite{jiang2023metalprognet}, and JPH-GBT~\cite{wang2025jphgbt}, the current state of the art for this task. Since PML is trained independently across 10 random seeds, we report both the average performance across all 10 runs and the performance of the single best-performing run, allowing direct comparison with prior work while also characterizing the stability of the model across different initializations.

As shown in Figure~\ref{fig:mpli-performance}(a) and Table~\ref{tab:mpli-comparison}, PML outperforms all baseline methods in terms of the PCC. Averaged over the 10 runs, PML achieves PCC $= 0.753 \pm 0.002$, with RMSE of $1.183\pm0.003$ kcal/mol. In its best-performing run (seed 121314), PML achieves PCC $= 0.756$, RMSE $= 1.176$ kcal/mol, surpassing the previously reported PCC of $0.745$ from CAML~\cite{feng2025caml}, the prior state of the art, and establishing a new state of the art for metalloprotein-ligand BA prediction. Figure~\ref{fig:mpli-performance}(b) shows the correlation between experimental and predicted BA for this best-performing run: the majority of points cluster closely around the diagonal, but predictions become noticeably more scattered and tend to underestimate the true affinity for the strongest-binding complexes (experimental BA below roughly $-10$ kcal/mol), a common pattern in binding affinity prediction that likely reflects the relative scarcity of high-affinity complexes in the training data.

\begin{figure}[!htbp]
    \centering
    \includegraphics[width=0.8\linewidth]{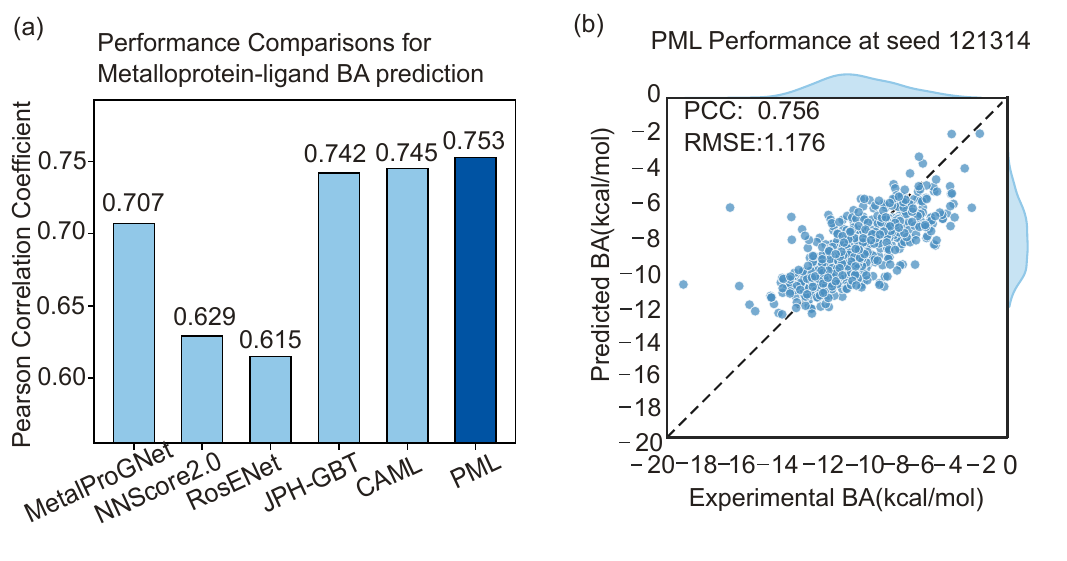}
    \caption{ (a) Performance comparison of PML with previously published models for metalloprotein-ligand binding affinity prediction, evaluated using the Pearson correlation coefficient. (b) Correlation between experimental and predicted binding affinities from PML's best-performing run (seed 121314), with PCC \(0.756\) kcal/mol and RMSE \(1.176\) kcal/mol.}
    \label{fig:mpli-performance}
\end{figure}
\begin{table}[!htbp]
\centering
\caption{Comparison of PML with previously published machine learning models on the metalloprotein-ligand binding affinity dataset~\cite{jiang2023metalprognet}. Baseline results for RosENet, NNScore2.0, and MetalProGNet are taken from~\cite{jiang2023metalprognet}, results for JPH-GBT are taken from~\cite{wang2025jphgbt}, and results for CAML(ES) are taken from~\cite{feng2025caml}. PML (average) reports the mean $\pm$ standard deviation over 10 independent runs with different random seeds, while PML (best run) reports the single best-performing run. RMSE values are computed on the raw $pK_d$ labels used as binding affinities.}
\label{tab:mpli-comparison}
\begin{tabular}{lcc}
\hline
Model & PCC & RMSE (kcal/mol) \\
\hline
RosENet~\cite{hassanharrirou2020rosenet} & $0.615\pm0.017$ & $1.436\pm0.011$ \\
NNScore2.0~\cite{durrant2011nnscore} & $0.629\pm0.002$ & $1.391\pm0.004$ \\
MetalProGNet~\cite{jiang2023metalprognet} & $0.703\pm0.010$ & $1.285\pm0.020$ \\
JPH-GBT~\cite{wang2025jphgbt} & $0.742\pm0.001$ & $1.205\pm0.001$ \\
CAML~\cite{feng2025caml} & $0.745\pm0.001$ & $1.202\pm0.002$ \\
PML (average over 10 runs) & $0.753\pm0.002$ & $1.183\pm0.003$ \\
PML (best run) & ${0.756}$ & ${1.176}$ \\
\hline
\end{tabular}

\end{table}

\subsubsection{Cross-Validation Performance}

To assess the sensitivity of PML on the training data, we further perform 10-fold cross-validation on the metalloprotein-ligand training set, again averaging over 10 independent random seeds to account for variability introduced by the stochastic subsampling in GBDT. Across all folds and seeds, PML achieves an average PCC of $0.759\pm0.003$, with an RMSE of $1.202\pm0.007$ kcal/mol. The best-performing run achieves a PCC of $0.761$ with an RMSE of $1.198$ kcal/mol, consistent with the held-out test set performance reported in the previous section and indicating that PML generalizes stably across different partitions of the training data.

To evaluate the contribution of the manifold embedding and protein language model features independent of the choice of the downstream regressor, we additionally train a support vector machine (SVM) as a basic model in \cite{sanchezgarcia2022ppiaffinity} on the same PML feature set and compare its performance against GBDT. As shown in Figure~\ref{fig:svm-gbdt-comparison}, GBDT substantially outperforms SVM on the same input features, achieving a PCC of $0.761$ and RMSE of $1.198$ kcal/mol, compared to a PCC of $0.689$ and RMSE of $1.340$ kcal/mol for SVM. This performance gap suggests that the relationship between the PML feature set and binding affinity is highly nonlinear, and that GBDT's ability to model complex, nonlinear feature interactions through an ensemble of decision trees is better suited to this task than the SVM's linear or kernel-based decision boundary. These results support the use of GBDT as the primary regressor throughout this work.

\begin{figure}[!htbp]
	\centering
	 \includegraphics[width=0.8\linewidth]{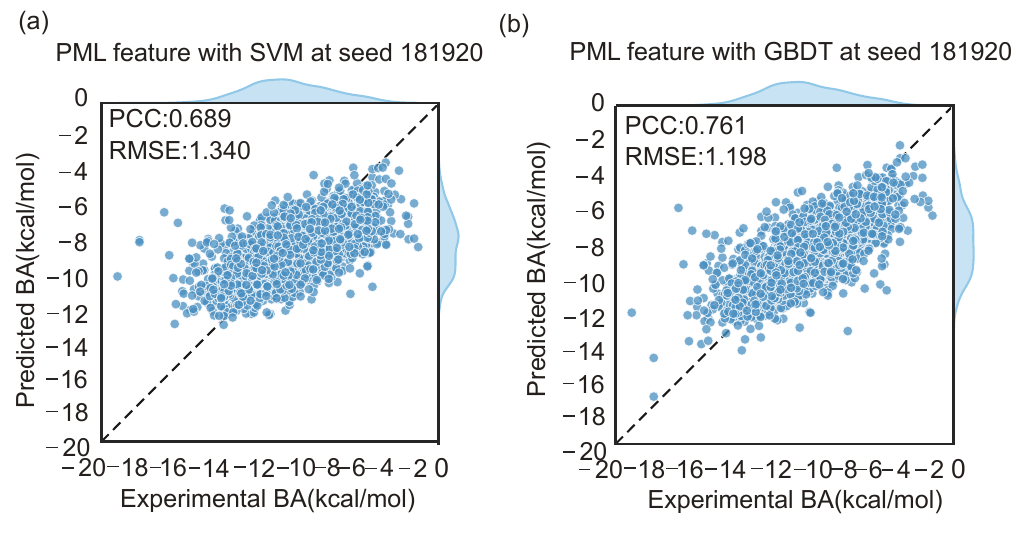}
	\caption{Comparison between SVM and GBDT regressors trained on the same PML feature set, for the metalloprotein-ligand binding affinity dataset (seed 181920). (a) Predicted versus experimental binding affinities using an SVM regressor. (b) Predicted versus experimental binding affinities using a GBDT regressor. Both panels show representative results from a single run.}
	\label{fig:svm-gbdt-comparison}
\end{figure}

\subsection{Performance of Protein-Protein Binding Affinity Predictions}

We next evaluate the proposed PML framework on the SKEMPI-WT dataset compare its performance against PLNet~\cite{xu2024pldtree}, a persistent-Laplacian-based neural network representing the current state of the art for protein-protein binding affinity prediction. To assess the contribution of the manifold embedding and protein language model features independently of the downstream regressor, we additionally evaluate a support vector machine trained on the same PML feature set, denoted PML-SVM, which mimics the modeling approach used in PPI-Affinity~\cite{sanchezgarcia2022ppiaffinity}, a widely used SVM-based baseline for protein-protein BA prediction. As with the metalloprotein-ligand task, PML with GBDT is trained across 10 independent random seeds, and we report both the average performance and the single best-performing run.

Table~\ref{tab:ppi-comparison} summarizes the performance of PML, PML-SVM, and PLNet on the SKEMPI-WT dataset. PML with GBDT consistently achieves the best performance, both on average across the 10 random seeds and in its best-performing run (seed 212223). Comparing PML against PML-SVM, which uses the same manifold embedding and protein language model features but a SVM as the downstream regressor, isolates the effect of the regressor while holding the feature set fixed: GBDT clearly outperforms SVM on identical input features, echoing the trend observed for the metalloprotein-ligand task and further supporting its suitability for modeling the nonlinear relationship between these features and binding affinity. More notably, PML also outperforms PLNet, a method built directly on persistent Laplacian and persistent homology descriptors of the binding interface. This suggests that the manifold embedding features introduced in this work, which combine geometric and topological information through the BIG Laplacian together with sequence-level protein language model embeddings, capture binding-relevant information beyond what persistent Laplacian and persistent homology features alone can provide, particularly for the flatter, less pocket-like interfaces characteristic of protein-protein complexes.

\begin{table}[t]
\centering
\caption{Comparison of PML with PLNet and with PML-SVM, an SVM-based ablation trained on the same PML feature set, on the SKEMPI-WT protein-protein binding affinity dataset. PML (average) and PML-SVM (average) report the mean $\pm$ standard deviation over 10 independent runs with different random seeds; PML (best run) and PML-SVM (best run) report the single best-performing run (seed 212223) for each model.}
\label{tab:ppi-comparison}
\begin{tabular}{lcc}
\hline
Model & PCC & RMSE (kcal/mol) \\
\hline
PLNet~\cite{xu2024pldtree} & $0.681\pm0.012$ & $1.533\pm0.021$ \\
PML-SVM (average over 10 runs) & $0.692\pm0.009$ & $2.127\pm0.025$ \\
PML-SVM (best run) & $0.705$ & $2.093$ \\
PML (average over 10 runs) & $0.713\pm0.016$ & $2.051\pm0.046$ \\
PML (best run) & ${0.731}$ & ${1.996}$ \\
\hline
\end{tabular}

\end{table}

\begin{figure}[t]
	\centering
	 \includegraphics[width=0.8\linewidth]{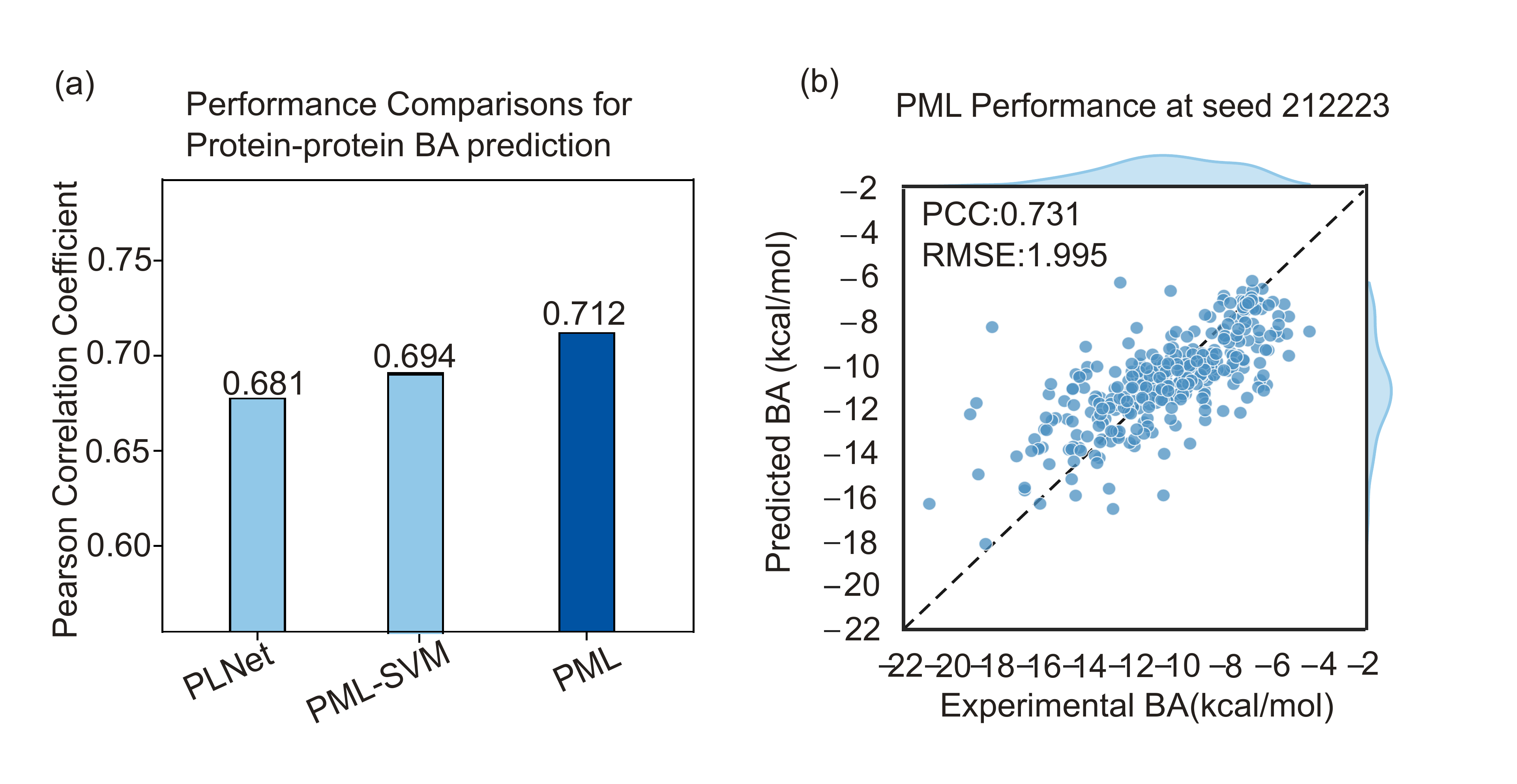}
	\caption{(a) Performance comparison of PML with PLNet and PML-SVM for protein-protein binding affinity prediction, evaluated using the Pearson correlation coefficient. (b) Correlation between experimental and predicted binding affinities from PML's best-performing run (seed 212223), with PCC $= 0.731$ and RMSE $= 1.996$ kcal/mol.}
	\label{fig:ppi-performance}
\end{figure}

\section{Discussion}\label{discussion}
We propose PML as a binding affinity prediction framework that unifies de Rham–Hodge theory, level-set filtrations, and machine learning. PML  outperforms existing methods on both metalloprotein-ligand and protein-protein benchmarks, despite the two interface types being governed by very different physics. Below we examine why a single manifold-based representation transfers across both settings, and why the spectral features are deliberately truncated.
\subsection{Element-specific manifolds and the metal coordination shell}
The element-specific decomposition is what makes a metal site legible to the manifold. Metal binding in proteins is mediated almost exclusively by nitrogen from histidine, sulfur from cysteine, and oxygen from aspartate and glutamate, held in tetrahedral or octahedral shells whose geometry is further tuned by second-shell residues~\cite{dudev2003first,shirian2021metalbinding}. These donors are exactly the protein-side element types $\{\text{N},\text{O},\text{S}\}$ retained in our scheme, and since a separate manifold is built for every atom-type pair, they are isolated rather than blurred into the surrounding bulk. The ion itself never enters the construction: it acts as a geometric constraint, organizing its donors into a rigid shell at short characteristic distances, and it is that organization the density field records. The $12\mathring{\text{A}}$ cutoff retains the second shell as well, so these pairs also carry the environment dictating the site's size, polarity, and selectivity.

The filtration reads that geometry out. Atoms held rigidly within a coordination shell overlap strongly and merge into a single connected component at the most negative isovalues, whereas loosely packed regions merge only much later. The isovalue at which $\beta_0$ drops is therefore a readout of local packing tightness, and tracking $\beta_0$ across the nine isovalues yields a multiscale signature of the binding site's organization that no single manifold contains.

\subsection{What the manifold captures at protein-protein interfaces}
Protein-protein interfaces present the opposite geometry: broad and flat rather than pocket-like, burying $1600\pm400$~\AA$^2$ of surface upon complexation~\cite{fleishman2014hotspot}. Affinity is dominated by shape complementarity and interfacial packing~\cite{kortemme2002simple}, concentrated in hot spots that are deeply buried and densely packed~\cite{deng2011dbac}, and correlates with the amount of surface buried~\cite{chen2013protein}. These are geometric quantities, and the filtration measures them: since each manifold spans atom pairs from both sides of the interface, the isovalue at which the two partners' density lobes merge reports how closely their surfaces approach, while the manifold volume reflects interfacial burial. The same descriptor that traces a coordination shell in a metalloprotein thus traces packing complementarity across a protein-protein interface, which is why one framework serves both.

Two differences are worth noting. Only $\{\text{C},\text{N},\text{O}\}$ is relevant here, so the manifold embedding contributes $6\times9\times9 = 486$ features against $2160$ for metalloprotein-ligand complexes, giving the ESM-2 embeddings proportionally greater weight. And SKEMPI-WT provides only 343 complexes, placing a premium on descriptor quality; that PML exceeds PLNet here, and that an SVM on the same features does too, indicates that the advantage lies in the representation rather than the regressor.

\subsection{Spectral truncation: why few nonzero eigenvalues}
Only $\beta_0$ and the first two nonzero eigenvalues of $L_{3,n}$ are retained per manifold. Extending to ten degrades performance on both benchmarks for two reasons. Numerically, the discrete spectrum approximates its continuum counterpart only up to a resolution-dependent cutoff: low modes vary across many grid cells and are resolved accurately, whereas high modes oscillate on a scale approaching the grid spacing $\ell$ and are contaminated by discretization error. Since the convergence of the BIG Laplacian spectrum~\cite{ribando2024graph} is not uniform, the highest modes describe the discretization as much as the manifold. Physically, higher eigenvalues correspond to finer surface detail, which the input data least reliably supports, being sensitive to coordinate uncertainty, the scale parameter $\tau$, and the single conformation of a static structure. The leading modes instead describe global shape and connectivity~\cite{fiedler1973algebraic}, which are robust to these perturbations. Spectral truncation is therefore a regularization matched to the resolution of the data, not a limitation to be lifted.

\section{Conclusion}\label{Conclusion}
As a foundational tool in differential geometry and algebraic topology, de Rham-Hodge theory studies differential forms, cohomology, and the harmonic structure of manifolds. Only recently has this framework been connected to data science, largely through evolutionary de Rham-Hodge theory~\cite{chen2021evolutionary} and its Eulerian, persistent extensions~\cite{su2024persistent}. This work builds on that direction by proposing Persistent Manifold Learning (PML), which represents a binding interface as a family of manifolds obtained through a level-set filtration of atomic density fields, and applies Boundary-Induced Graph (BIG) Laplacians to capture both the topology and the local geometric organization of the interface across scales. We pair these manifold embedding features with protein and molecular language model embeddings, which supply complementary sequence-level information, and combine both with GBDT, well suited to this setting for its accuracy and efficiency on relatively small, structured biomolecular datasets.

We validate PML on two structurally distinct benchmarks: a metalloprotein-ligand binding dataset and the wild-type subset of SKEMPI v2. Because the two interaction types differ in chemical composition and interface geometry, we introduce an element-specific atom-pairing scheme tailored to each, allowing a single framework to accommodate both compact, metal-coordinated pockets and broad, flat protein-protein interfaces. Compared against state-of-the-art methods in the literature, PML achieves the best reported performance on both benchmarks. Binding affinity prediction serves here as a case study; the framework extends in principle to other biomolecular problems in which structural and sequence information jointly determine a quantitative outcome. We believe that manifold-based, de Rham-Hodge-theoretic learning represents a promising emerging direction in machine learning for molecular and biological data.

\section*{Further Reading}

The mathematical foundations of this work lie in de Rham-Hodge theory and its discretization. The following references are intended as entry points for readers approaching the article from the life sciences side.

\medskip
\noindent For background on differential forms, de Rham cohomology, and Hodge theory, see the following:
\begin{itemize}
\item R. Bott and L. W. Tu, \emph{Differential Forms in Algebraic Topology}, Graduate Texts in Mathematics 82, Springer, 1982. A geometric introduction to differential forms and de Rham cohomology, requiring little prior topology.
\item G. Schwarz, \emph{Hodge Decomposition—A Method for Solving Boundary Value Problems}, Lecture Notes in Mathematics 1607, Springer, 1995. The standard reference for Hodge theory on manifolds \emph{with boundary}, covering the normal and tangential boundary conditions and the harmonic field decompositions used in Section~\ref{sec.deRhamHodge}.
\end{itemize}

\medskip
\noindent For details on discretization and the Laplacian-based descriptors used here, see the following:
\begin{itemize}
\item E. Ribando-Gros, R. Wang, J. Chen, Y. Tong, and G.-W. Wei, \emph{Combinatorial and Hodge Laplacians: Similarities and differences}, SIAM Rev., (2024). Introduces the Boundary-Induced Graph Laplacian used in this article.
\item J. Chen, R. Zhao, Y. Tong, and G.-W. Wei, \emph{Evolutionary de Rham--Hodge method}, Discrete Contin. Dyn. Syst. B, 26 (2021), pp. 3785--3821. Extends de Rham--Hodge theory to families of evolving manifolds, the multiscale setting adopted here.
\item Z. Su, Y. Tong, and G.-W. Wei, \emph{Persistent de Rham--Hodge Laplacians in Eulerian representation for manifold topological learning}, AIMS Math., 9 (2024), pp. 27438--27470. The Eulerian, grid-based formulation on which our implementation is built.
\end{itemize}

\section*{Declarations}

\section*{Authorship and Contributorship}

X.X. implemented the framework, performed all computations, and wrote the original draft. Z. S. contributed to the design of the method, assisted with the implementation, and revised the manuscript. G.-W. W and C. W.  conceived and supervised the project, secured funding, and revised the manuscript. All authors read and approved the final manuscript.

\section*{Conflicts of interest}  The authors declare no conflicts of interest.

\section*{Data \& Code Availability}
The code needed to reproduce this paper's results is available from the authors upon reasonable request. For detailed information on the metalloprotein-ligand complex dataset, please refer to~\cite{wang2025join}. The SKEMPI v2 dataset is available at \url{https://life.bsc.es/pid/skempi2/}.

\section*{Ethics} 
 This study did not involve human participants, animals, or sensitive data requiring
ethical approval or consent to participate.
Artificial Intelligence was used for English editing.

\section*{Acknowledgments} C. W. was partially supported  by the National Science Foundation under award DMS-2206332.
G.-W. W. was supported in part by NIH grant R35GM148196 and Georgia Research Alliance.

\bibliographystyle{plain}
\bibliography{main}
\end{document}

%% file: ex_shared.tex
\newsiamremark{remark}{Remark}
\newsiamremark{hypothesis}{Hypothesis}
\crefname{hypothesis}{Hypothesis}{Hypotheses}
\newsiamthm{claim}{Claim}
\newsiamthm{assum}{Assumption}
\newsiamthm{scheme}{Scheme}
\newsiamthm{lem}{Lemma}
\newsiamthm{thm}{Theorem}
\headers{Persistent Manifold Learning }{X.~Xu, Z.~Su, G.~Wei, C.~Wang}

\title{Persistent Manifold Learning of   Protein Properties}

\author{
Xingjian Xu\thanks{
Department of Mathematics, University of Florida, Gainesville, FL 32611, USA 
({\tt xingjianxu@ufl.edu}).}
\and
Zhe Su\thanks{
Department of Mathematics and Statistics, Auburn University, Auburn, AL 36849, USA 
({\tt zhs0011@auburn.edu}).}
\and
Guo-Wei Wei\thanks{
  Department of Mathematics,   University of Georgia, Athens, GA 30602, USA.\\
   Department of Biochemistry and Molecular Biology, University of Georgia, Athens, GA 30602, USA
({\tt guowei.wei@uga.edu  }).}
\and
Chunmei Wang\thanks{
Department of Mathematics, University of Florida, Gainesville, FL 32611, USA 
(corresponding author: {\tt chunmei.wang@ufl.edu}).}
}



%% file: main.bib
@article{liu2024algebraic,
  title={The algebraic stability for persistent Laplacians},
  author={Liu, Jian and Li, Jingyan and Wu, Jie},
  journal={Homology, Homotopy and Applications},
  volume={26},
  number={2},
  pages={297--323},
  year={2024},
  publisher={International Press of Boston}
}

@article{townsend2020representation,
  title={Representation of molecular structures with persistent homology for machine learning applications in chemistry},
  author={Townsend, Jacob and Micucci, Cassie Putman and Hymel, John H and Maroulas, Vasileios and Vogiatzis, Konstantinos D},
  journal={Nature communications},
  volume={11},
  number={1},
  pages={3230},
  year={2020},
  publisher={Nature Publishing Group UK London}
}

@article{su2025topological,
  title={Topological data analysis and topological deep learning beyond persistent homology: a review},
  author={Su, Zhe and Liu, Xiang and Hamdan, Layal Bou and Maroulas, Vasileios and Wu, Jie and Carlsson, Gunnar and Wei, Guo-Wei},
  journal={Artificial Intelligence Review},
  volume  = {59},
  pages   = {58},
  year={2025},
  publisher={Springer}
}

@article{liu2026manifold,
  title={Manifold topological deep learning for biomedical data},
  author={Liu, Xiang and Su, Zhe and Shi, Yongyi and Tong, Yiying and Wang, Ge and Wei, Guo-Wei},
  journal={Nature Communications},
   volume  = {17},
  pages   = {4710},
  year={2026},
  publisher={Nature Publishing Group UK London}
}

@article{wei2025persistent,
  title={Persistent sheaf laplacians},
  author={Wei, Xiaoqi and Wei, Guo-Wei},
  journal={Foundations of data science (Springfield, Mo.)},
  volume={7},
  number={2},
  pages={446},
  year={2025}
}

@article{lu2020protein,
  title   = {Recent Advances in the Development of {Protein--Protein Interaction} Modulators: Mechanisms and Clinical Trials},
  author  = {Haiying Lu and Qiaodan Zhou and Jun He and Zhongliang Jiang and Cheng Peng and Rongsheng Tong and Jianyou Shi},
  journal = {Signal Transduction and Targeted Therapy},
  volume  = {5},
  pages   = {213},
  year    = {2020},
  doi     = {10.1038/s41392-020-00315-3}
}

@article{du2016protein,
  title   = {Insights into {Protein--Ligand} Interactions: Mechanisms, Models, and Methods},
  author  = {Xiaoyong Du and Youhua Li and Yao-Liang Xia and Shi-Meng Ai and Jie Liang and Peng Sang and Xing-Lai Ji and Shihua Liu},
  journal = {International Journal of Molecular Sciences},
  volume  = {17},
  number  = {2},
  pages   = {144},
  year    = {2016},
  doi     = {10.3390/ijms17020144}
}

@article{zhao2022brief,
  title   = {A Brief Review of {Protein--Ligand} Interaction Prediction},
  author  = {Lingling Zhao and Hao L. Ciallella and Lauren M. Aleksunes and Hao Zhu},
  journal = {Computational and Structural Biotechnology Journal},
  volume  = {20},
  pages   = {2831--2838},
  year    = {2022},
  doi     = {10.1016/j.csbj.2022.06.004}
}

@article{nada2024protein,
  title   = {New Insights into {Protein--Protein Interaction} Modulators in Drug Discovery and Therapeutic Advance},
  author  = {Hossam Nada and Yongseok Choi and Sungdo Kim and Kwon Su Jeong and Nicholas A. Meanwell and Kyeong Lee},
  journal = {Signal Transduction and Targeted Therapy},
  volume  = {9},
  pages   = {341},
  year    = {2024},
  doi     = {10.1038/s41392-024-02036-3}
}

@article{bennett2022protein,
  title   = {{Protein--Small Molecule} Interactions in Native Mass Spectrometry},
  author  = {Jack L. Bennett and Giang T. H. Nguyen and William A. Donald},
  journal = {Chemical Reviews},
  volume  = {122},
  number  = {8},
  pages   = {7327--7385},
  year    = {2022},
  doi     = {10.1021/acs.chemrev.1c00293}
}

@article{scott2016small,
  title   = {Small Molecules, Big Targets: Drug Discovery Faces the {Protein--Protein} Interaction Challenge},
  author  = {Duncan E. Scott and Andrew R. Bayly and Chris Abell and John Skidmore},
  journal = {Nature Reviews Drug Discovery},
  volume  = {15},
  number  = {8},
  pages   = {533--550},
  year    = {2016},
  doi     = {10.1038/nrd.2016.29}
}

@article{gilson1997statistical,
  title   = {The Statistical-Thermodynamic Basis for Computation of Binding Affinities: A Critical Review},
  author  = {Michael K. Gilson and James A. Given and Bruce L. Bush and J. Andrew McCammon},
  journal = {Biophysical Journal},
  volume  = {72},
  number  = {3},
  pages   = {1047--1069},
  year    = {1997},
  doi     = {10.1016/S0006-3495(97)78756-3}
}

@article{kitchen2004docking,
  title   = {Docking and Scoring in Virtual Screening for Drug Discovery: Methods and Applications},
  author  = {Douglas B. Kitchen and H{\'e}l{\`e}ne Decornez and John R. Furr and J{\"u}rgen Bajorath},
  journal = {Nature Reviews Drug Discovery},
  volume  = {3},
  number  = {11},
  pages   = {935--949},
  year    = {2004},
  doi     = {10.1038/nrd1549}
}

@article{vangone2015contacts,
  title   = {Contacts-Based Prediction of Binding Affinity in {Protein--Protein} Complexes},
  author  = {Anna Vangone and Alexandre M. J. J. Bonvin},
  journal = {eLife},
  volume  = {4},
  pages   = {e07454},
  year    = {2015},
  doi     = {10.7554/eLife.07454}
}

@article{erijman2014structure,
  title   = {How Structure Defines Affinity in {Protein--Protein} Interactions},
  author  = {Ariel Erijman and Elad Rosenthal and Julia M. Shifman},
  journal = {PLOS ONE},
  volume  = {9},
  number  = {10},
  pages   = {e110085},
  year    = {2014},
  doi     = {10.1371/journal.pone.0110085}
}

@article{feinberg2018potentialnet,
  title={PotentialNet for molecular property prediction},
  author={Feinberg, Evan N. and Sur, Debnil and Wu, Zhenqin and Husic, Brooke E. and Mai, Huanghao and Li, Yang and Sun, Saisai and Yang, Jianyi and Ramsundar, Bharath and Pande, Vijay S.},
  journal={ACS Central Science},
  volume={4},
  number={11},
  pages={1520--1530},
  year={2018},
  doi={10.1021/acscentsci.8b00507}
}

@article{nguyen2021graphdta,
  title={GraphDTA: Predicting drug--target binding affinity with graph neural networks},
  author={Nguyen, Thin and Le, Hang and Quinn, Thomas P. and Nguyen, Tri and Le, Truyen and Venkatesh, Svetha},
  journal={Bioinformatics},
  volume={37},
  number={8},
  pages={1140--1147},
  year={2021},
  doi={10.1093/bioinformatics/btaa921}
}

@article{li2020monn,
  title={MONN: A multi-objective neural network for predicting compound--protein interactions and affinities},
  author={Li, Shuya and Wan, Fangping and Shu, Hantao and Jiang, Tao and Zhao, Dan and Zeng, Jianyang},
  journal={Cell Systems},
  volume={10},
  number={4},
  pages={308--322.e11},
  year={2020},
  doi={10.1016/j.cels.2020.03.002}
}

@article{jiang2021interactiongraphnet,
  title={InteractionGraphNet: A novel and efficient deep graph representation learning framework for accurate protein--ligand interaction predictions},
  author={Jiang, Dejun and Hsieh, Chang-Yu and Wu, Zhenxing and Kang, Yu and Wang, Jike and Wang, Ercheng and Liao, Ben and Shen, Chao and Xu, Lei and Wu, Jian and Cao, Dejun and Hou, Tingjun},
  journal={Journal of Medicinal Chemistry},
  year={2021},
  doi={10.1021/acs.jmedchem.1c01830}
}

@article{xia2014persistent,
  title={Persistent homology analysis of protein structure, flexibility, and folding},
  author={Xia, Kelin and Wei, Guo-Wei},
  journal={International Journal for Numerical Methods in Biomedical Engineering},
  volume={30},
  number={8},
  pages={814--844},
  year={2014},
  doi={10.1002/cnm.2655}
}

@article{cang2017topologynet,
  title={TopologyNet: Topology based deep convolutional and multi-task neural networks for biomolecular property predictions},
  author={Cang, Zixuan and Wei, Guo-Wei},
  journal={PLOS Computational Biology},
  volume={13},
  number={7},
  pages={e1005690},
  year={2017},
  doi={10.1371/journal.pcbi.1005690}
}

@article{cang2018integration,
  title={Integration of element specific persistent homology and machine learning for protein--ligand binding affinity prediction},
  author={Cang, Zixuan and Wei, Guo-Wei},
  journal={International Journal for Numerical Methods in Biomedical Engineering},
  volume={34},
  number={2},
  pages={e2914},
  year={2018},
  doi={10.1002/cnm.2914}
}

@article{cang2018representability,
  title={Representability of algebraic topology for biomolecules in machine learning based scoring and virtual screening},
  author={Cang, Zixuan and Mu, Lin and Wei, Guo-Wei},
  journal={PLOS Computational Biology},
  volume={14},
  number={1},
  pages={e1005929},
  year={2018},
  doi={10.1371/journal.pcbi.1005929}
}

@article{long2025interpretable,
  title={Interpretable binding affinity prediction with persistent homology},
  author={Long, Y. and others},
  journal={PLOS Computational Biology},
  volume={21},
  number={6},
  pages={e1013216},
  year={2025},
  doi={10.1371/journal.pcbi.1013216}
}

@article{wang2020persistent,
  title={Persistent spectral graph},
  author={Wang, Rui and Nguyen, Duc Duy and Wei, Guo-Wei},
  journal={International Journal for Numerical Methods in Biomedical Engineering},
  volume={36},
  number={9},
  pages={e3376},
  year={2020},
  doi={10.1002/cnm.3376}
}

@article{memoli2022persistent,
  title={Persistent Laplacians: Properties, algorithms and implications},
  author={M{\'e}moli, Facundo and Wan, Zhengchao and Wang, Yusu},
  journal={SIAM Journal on Mathematics of Data Science},
  volume={4},
  number={2},
  pages={858--884},
  year={2022},
  doi={10.1137/21M1435471}
}

@article{wang2023persistentpath,
  title={Persistent path Laplacian},
  author={Wang, Rui and others},
  journal={Foundations of Data Science},
  volume={5},
  number={1},
  pages={26--55},
  year={2023},
  doi={10.3934/fods.2022015}
}

@article{meng2021persistent,
title={Persistent spectral-based machine learning ({PerSpect ML}) for protein--ligand binding affinity prediction},
author={Meng, Zhenyu and Xia, Kelin},
journal={Science Advances},
volume={7},
number={19},
pages={eabc5329},
year={2021},
doi={10.1126/sciadv.abc5329}
}

@article{chen2022persistent,
title={Persistent Laplacian projected Omicron BA.4 and BA.5 to become new dominating variants},
author={Chen, Jiahui and Qiu, Yuchi and Wang, Rui and Wei, Guo-Wei},
journal={Computers in Biology and Medicine},
volume={151},
pages={106262},
year={2022},
doi={10.1016/j.compbiomed.2022.106262}
}

@article{xu2024pldtree,
title={{PLNet}: Persistent Laplacian neural network for protein--protein binding free energy prediction},
author={Xu, Xingjian and Wang, Chunmei and Wei, Guo-Wei and Chen, Jiahui},
journal={Protein Science},
volume={34},
number={12},
pages={e70377},
year={2025},
doi={10.1002/pro.70377}
}

@article{moesser2022protein,
  title={Protein--ligand interaction graphs: Learning from ligand-shaped 3D interaction graphs to improve binding affinity prediction},
  author={Moesser, Marc A. and Klein, Dominik K. and Boyles, Fergus and Deane, Charlotte M. and Baxter, Alice and Morris, Garrett M.},
  journal={bioRxiv},
  year={2022},
  doi={10.1101/2022.03.04.483012}
}

@article{chen2021evolutionary,
  title={Evolutionary de {Rham}--{Hodge} method},
  author={Chen, Jiahui and Zhao, Rundong and Tong, Yiying and Wei, Guo-Wei},
  journal={Discrete and Continuous Dynamical Systems - Series B},
  volume={26},
  number={7},
  pages={3785--3821},
  year={2021},
  doi={10.3934/dcdsb.2020257}
}

@article{su2024persistent,
  title={Persistent de {Rham}--{Hodge} Laplacians in {Eulerian} representation for manifold topological learning},
  author={Su, Zhe and Tong, Yiying and Wei, Guo-Wei},
  journal={AIMS Mathematics},
  volume={9},
  number={10},
  pages={27438--27470},
  year={2024},
  doi={10.3934/math.20241333}
}

@article{su2024topology,
  title={Topology-preserving Hodge decomposition in the Eulerian representation},
  author={Su, Zhe and Tong, Yiying and Wei, Guo-Wei},
  journal={Beijing Journal of Pure and Applied Mathematics},
  volume={2},
  number={2},
  pages={619--657},
  year={2025},
  publisher={International Press of Boston}
}

@article{jiang2023metalprognet,
  title={MetalProGNet: A structure-based deep graph model for metalloprotein--ligand interaction predictions},
  author={Jiang, Dejun and Ye, Zixuan and Hsieh, Chang-Yu and Yang, Ziyi and Zhang, Xiaohui and Kang, Yu and Du, Haitao and Wu, Zhenxing and Wang, Jike and Zeng, Yi},
  journal={Chemical Science},
  volume={14},
  number={8},
  pages={2054--2069},
  year={2023},
  doi={10.1039/D2SC06576B}
}

@article{jankauskaite2019skempi,
  title={{SKEMPI} 2.0: An updated benchmark of changes in protein--protein binding energy, kinetics and thermodynamics upon mutation},
  author={Jankauskait{\.e}, Justina and Jim{\'e}nez-Garc{\'i}a, Brian and Dapkunas, Justas and Fern{\'a}ndez-Recio, Juan and Moal, Iain H.},
  journal={Bioinformatics},
  volume={35},
  number={3},
  pages={462--469},
  year={2019},
  doi={10.1093/bioinformatics/bty635}
}

@article{rives2021biological,
  title={Biological structure and function emerge from scaling unsupervised learning to 250 million protein sequences},
  author={Rives, Alexander and Meier, Joshua and Sercu, Tom and Goyal, Siddharth and Lin, Zeming and Liu, Jason and Guo, Demi and Ott, Myle and Zitnick, C. Lawrence and Ma, Jerry and Fergus, Rob},
  journal={Proceedings of the National Academy of Sciences},
  volume={118},
  number={15},
  pages={e2016239118},
  year={2021}
}

@article{lin2023evolutionary,
  title={Evolutionary-scale prediction of atomic-level protein structure with a language model},
  author={Lin, Zeming and Akin, Halil and Rao, Roshan and Hie, Brian and Zhu, Zhongkai and Lu, Wenting and dos Santos Costa, Allan and Fazel-Zarandi, Maryam and Sercu, Tom and Candido, Salvatore and others},
  journal={Science},
  volume={379},
  number={6637},
  pages={1123--1130},
  year={2023}
}

@article{nguyen2019dg,
  author  = {Nguyen, Duc Duy and Wei, Guo-Wei},
  title   = {{DG-GL}: Differential geometry-based geometric learning of molecular datasets},
  journal = {International Journal for Numerical Methods in Biomedical Engineering},
  volume  = {35},
  number  = {3},
  pages   = {e3179},
  year    = {2019},
  doi     = {10.1002/cnm.3179}
}

@article{friedrichs1955differential,
  author  = {Friedrichs, K. O.},
  title   = {Differential forms on {R}iemannian manifolds},
  journal = {Communications on Pure and Applied Mathematics},
  volume  = {8},
  number  = {4},
  pages   = {551--590},
  year    = {1955},
  doi     = {10.1002/cpa.3160080408}
}

@incollection{desbrun2006discrete,
  author    = {Desbrun, Mathieu and Kanso, Eva and Tong, Yiying},
  title     = {Discrete differential forms for computational modeling},
  booktitle = {ACM SIGGRAPH 2006 Courses},
  pages     = {39--54},
  year      = {2006},
  publisher = {ACM},
  doi       = {10.1145/1185657.1185665}
}

@article{ribando2024graph,
  author  = {Ribando-Gros, Emily and Wang, Rui and Chen, Jiahui and Tong, Yiying and Wei, Guo-Wei},
  title   = {Combinatorial and {H}odge {L}aplacians: Similarities and Differences},
  journal = {SIAM Review},
  year    = {2024},
  doi     = {10.1137/22M1482299}
}

@book{hodge1941theory,
  author    = {Hodge, William V. D.},
  title     = {The Theory and Applications of Harmonic Integrals},
  publisher = {Cambridge University Press},
  year      = {1941}
}

@article{fiedler1973algebraic,
  author  = {Fiedler, Miroslav},
  title   = {Algebraic connectivity of graphs},
  journal = {Czechoslovak Mathematical Journal},
  volume  = {23},
  number  = {2},
  pages   = {298--305},
  year    = {1973}
}

@article{wang2025join,
  title={Join persistent homology (jph)-based machine learning for metalloprotein--ligand binding affinity prediction},
  author={Wang, Yaxing and Liu, Xiang and Zhang, Yipeng and Wang, Xiangjun and Xia, Kelin},
  journal={Journal of Chemical Information and Modeling},
  volume={65},
  number={6},
  pages={2785--2793},
  year={2025},
  publisher={ACS Publications}
}

@article{xiang2001extending,
  author  = {Xiang, Zhexin and Honig, Barry},
  title   = {Extending the accuracy limits of prediction for side-chain conformations},
  journal = {Journal of Molecular Biology},
  volume  = {311},
  number  = {2},
  pages   = {421--430},
  year    = {2001},
  doi     = {10.1006/jmbi.2001.4865}
}

@book{schwarz2006hodge,
  author    = {Schwarz, G{\"u}nter},
  title     = {Hodge Decomposition: A Method for Solving Boundary Value Problems},
  series    = {Lecture Notes in Mathematics},
  volume    = {1607},
  publisher = {Springer},
  address   = {Berlin},
  year      = {1995},
  doi       = {10.1007/BFb0095978}
}

@article{chithrananda2020chemberta,
  author  = {Chithrananda, Seyone and Grand, Gabriel and Ramsundar, Bharath},
  title   = {{ChemBERTa}: Large-Scale Self-Supervised Pretraining for Molecular Property Prediction},
  journal = {arXiv preprint arXiv:2010.09885},
  year    = {2020}
}

@article{wang2025jphgbt,
  author  = {Wang, Yaxing and Liu, Xiang and Zhang, Yipeng and Wang, Xiangjun and Xia, Kelin},
  title   = {Join Persistent Homology ({JPH})-Based Machine Learning for Metalloprotein--Ligand Binding Affinity Prediction},
  journal = {Journal of Chemical Information and Modeling},
  year    = {2025}
}

@article{hassanharrirou2020rosenet,
  author  = {Hassan-Harrirou, Hussein and Zhang, Ce and Lemmin, Thomas},
  title   = {{RosENet}: Improving Binding Affinity Prediction by Leveraging Molecular Mechanics Energies with an Ensemble of {3D} Convolutional Neural Networks},
  journal = {Journal of Chemical Information and Modeling},
  volume  = {60},
  number  = {6},
  pages   = {2791--2802},
  year    = {2020},
  doi     = {10.1021/acs.jcim.0c00075}
}

@article{durrant2011nnscore,
  author  = {Durrant, Jacob D. and McCammon, J. Andrew},
  title   = {{NNScore} 2.0: A Neural-Network Receptor--Ligand Scoring Function},
  journal = {Journal of Chemical Information and Modeling},
  volume  = {51},
  number  = {11},
  pages   = {2897--2903},
  year    = {2011},
  doi     = {10.1021/ci2003889}
}

@article{sanchezgarcia2022ppiaffinity,
  author  = {S{\'a}nchez-Garc{\'i}a, Ruben and Sorzano, Carlos Oscar S. and Carazo, Jos{\'e} Mar{\'i}a and Segura, Joan},
  title   = {{PPI-Affinity}: A Web Tool for the Prediction and Optimization of Protein--Peptide and Protein--Protein Binding Affinity},
  journal = {Journal of Proteome Research},
  volume  = {21},
  number  = {8},
  pages   = {1829--1841},
  year    = {2022},
  doi     = {10.1021/acs.jproteome.2c00020}
}

@article{feng2025caml,
  author  = {Feng, Hongsong and Suwayyid, Faisal and Zia, Mushal and Wee, JunJie and Hozumi, Yuta and Chen, Chun-Long and Wei, Guo-Wei},
  title   = {{CAML}: Commutative Algebra Machine Learning---A Case Study on Protein--Ligand Binding Affinity Prediction},
  journal = {Journal of Chemical Information and Modeling},
  year    = {2025},
  doi     = {10.1021/acs.jcim.5c00940}
}

@article{dudev2003first,
  author  = {Dudev, Todor and Lin, Yen-lin and Dudev, Milena and Lim, Carmay},
  title   = {First--Second Shell Interactions in Metal Binding Sites in Proteins: A {PDB} Survey and {DFT/CDM} Calculations},
  journal = {Journal of the American Chemical Society},
  volume  = {125},
  number  = {10},
  pages   = {3168--3180},
  year    = {2003},
  doi     = {10.1021/ja0209722}
}

@article{shirian2021metalbinding,
  author  = {Permyakov, Eugene A.},
  title   = {Metal Binding Proteins},
  journal = {Encyclopedia},
  volume  = {1},
  number  = {1},
  pages   = {261--292},
  year    = {2021},
  doi     = {10.3390/encyclopedia1010024}
}

@article{fleishman2014hotspot,
  author  = {Fleishman, Sarel J. and Baker, David},
  title   = {Hotspot-Centric De Novo Design of Protein Binders},
  journal = {Journal of Molecular Biology},
  volume  = {413},
  number  = {5},
  pages   = {1047--1062},
  year    = {2011},
  doi     = {10.1016/j.jmb.2011.09.001}
}

@article{kortemme2002simple,
  author  = {Kortemme, Tanja and Baker, David},
  title   = {A Simple Physical Model for Binding Energy Hot Spots in Protein--Protein Complexes},
  journal = {Proceedings of the National Academy of Sciences},
  volume  = {99},
  number  = {22},
  pages   = {14116--14121},
  year    = {2002},
  doi     = {10.1073/pnas.202485799}
}

@article{deng2011dbac,
  author  = {Deng, Lei and Guan, Jihong and Wei, Xiaoming and Yi, Yuan and Zhang, Qiwen and Zhou, Shuigeng},
  title   = {{DBAC}: A Simple Prediction Method for Protein Binding Hot Spots Based on Burial Levels and Deeply Buried Atomic Contacts},
  journal = {BMC Systems Biology},
  volume  = {5},
  number  = {Suppl 1},
  pages   = {S5},
  year    = {2011},
  doi     = {10.1186/1752-0509-5-S1-S5}
}

@article{chen2013protein,
  author  = {Chen, Jing and Sawyer, Nicholas and Regan, Lynne},
  title   = {Protein--Protein Interactions: General Trends in the Relationship Between Binding Affinity and Interfacial Buried Surface Area},
  journal = {Protein Science},
  volume  = {22},
  number  = {4},
  pages   = {510--515},
  year    = {2013},
  doi     = {10.1002/pro.2230}
}
